\documentclass[11pt]{article}
\usepackage{amsfonts}
\usepackage{amsmath}
\usepackage{geometry}
\usepackage{graphicx}
\usepackage{rotating}
\usepackage{subcaption}
\usepackage{caption}

\newtheorem{theorem}{Theorem}

\newenvironment{proof}[1][Proof]{\noindent\textbf{#1.} }{\ \rule{0.5em}{0.5em}}
\input{tcilatex}
\begin{document}

\title{Incorporating Social Welfare in Program-Evaluation and Treatment
Choice\thanks{%
We are grateful to Meredith Crowley, Peter Hammond, Arthur Lewbel and Albert
Park for discussions related to the topic of this paper.}}
\author{Debopam Bhattacharya\thanks{%
Bhattacharya acknowledges financial support from the European Research
Council via a Consolidator Grant EDWEL, Project number 681565.} \\
University Of Cambridge\thanks{%
Address for correspondence: debobhatta@gmail.com} \and Tatiana Komarova \\
University of Manchester}
\date{October 15, 2022}
\maketitle

\begin{abstract}
The econometric literature on treatment-effects typically takes functionals
of outcome-distributions as `social welfare' and ignores program-impacts on
unobserved utilities. We show how to incorporate aggregate utility within
econometric program-evaluation and optimal treatment-targeting for a
heterogenous population. In the practically important setting of
discrete-choice, under unrestricted preference-heterogeneity and
income-effects, the indirect-utility distribution becomes a closed-form
functional of average demand. This enables nonparametric cost-benefit
analysis of policy-interventions and their optimal targeting based on
planners' redistributional preferences. For ordered/continuous choice,
utility-distributions can be bounded. Our methods are illustrated with
Indian survey-data on private-tuition, where income-paths of
usage-maximizing subsidies differ significantly from welfare-maximizing ones.


JEL Codes C14 C25 D12 D31 D61 D63
\end{abstract}

\pagebreak

\section{Introduction}

Data-driven program evaluation has become key to modern policy analysis, and
has produced a large body of research on treatment-effect analysis cf.
Heckman-Vytlacil 2007, Imbens-Wooldridge 2009, and on optimal treatment
choice cf. Manski 2004. Both these literatures have exclusively focused on
functionals of\emph{\ outcome }distributions as the key object of interest,
and bypassed the classic public-economics question of measuring
program-effects on unobservable\emph{\ utilities} of individuals. For
example, a college financial aid program would typically be evaluated in the
econometric tradition via its `treatment effect' on aggregate enrolment or
future earnings etc., and the treatment-choice problem would address the
question of how to target a limited amount of subsidy funds to maximize this
aggregate cf. Bhattacharya and Dupas 2012, Kitagawa and Tetenov 2018 etc.
This approach ignores the question of how much and how differently do the
individuals themselves value the subsidy -- determined by their
willingness-to-pay for college -- and what is the subsidy's effect on
aggregate utility, weighted by the social planner's distributional
preferences. In particular, those subsidy-eligibles who would attend college
irrespective of the subsidy would contribute nothing to the average
treatment-effect, although their savings from subsidized tuition would raise
their utility without changing their attendance behavior. Secondly, in many
practical settings, multiple related outcomes of policy-interest are likely
to be affected by a single intervention. For example, a price subsidy for
mosquito-nets (cf. Dupas 2014) can be evaluated in terms of aggregate
take-up of nets, incidence of malaria, school absence of children and so
forth. The aggregate utility approach provides a natural way to aggregate
these separate effects via how they determine the households' overall
willingness-to-pay for the mosquito-net. Third, price-interventions for
redistributions are often politically motivated. Efficiency-costs of such
policies to society are therefore important metrics of political assessment.

Indeed, in public economics, cost-benefit analysis of an intervention was
traditionally conducted by comparing the expenditure on it with the change
it brings about in aggregate indirect-utility, cf. Bergson 1938, Samuelson
1947 and Mirrlees 1971. However, when bringing these concepts to data, this
literature ignored unobserved heterogeneity and imposed arbitrary
functional-form restrictions on the utility functions of `representative'
consumers who were assumed to vary solely in terms of observables cf. Deaton
1984 and Ahmad and Stern 1984. Later work such as Feldstein 1999, Saez 2001
have shown that in labour-supply models with consumption-leisure trade-off
and heterogeneous agents, the optimal income-tax rate depends on individual
heterogeneity via certain aggregates only, such as the average elasticity of
taxable income w.r.t. the marginal tax-rate, where the taxable income
distribution is endogenously determined with the tax-rate. For empirical
implementation, these results require parametric modelling of preferences:
cf. Saez 2001 page 219 and Section 5. Similarly, Manski 2014 derives bounds
on optimal income-tax schedule when consumers have heterogeneous
Cobb-Douglas preferences. In an econometric sense, these are not
nonparametric `identification' results that express the object of interest
(consumer welfare, optimal tax-schedule etc.) in terms of quantities
directly estimable from the data without making unjustified functional-form
assumptions about unobservables.

The present paper shows that in the practically important setting of
multinomial choice, the distribution of consumers' indirect utilities
induced by unobserved heterogeneity, can be expressed as closed-form
functionals of choice-probability functions. This result assumes no
knowledge of functional-form of utilities, nature of income-effects and
dimension/distribution of unobserved preference-heterogeneity. Thus, solely
the knowledge of choice-probabilities enables fully nonparametric evaluation
of interventions and design of optimal treatment-choice based on aggregate
utility. The knowledge of the indirect utility distribution further permits
measurement of efficiency-loss required to ensure a desired average outcome;
for instance, in our tuition-subsidy example above, one can calculate the
monetary value of the subsidy-induced market distortion (excess-burden)
necessary to reach an enrolment target of say 80\%, thus providing a
theoretically justified numerical measure of the equity-efficiency trade-off
involved. The above exercise cannot be performed entirely nonparametrically
in ordered/continuous choice settings; but sharp bounds on
welfare-distributions can be calculated there using our result.

An alternative way to evaluate welfare is via expenditure-function based
Hicksian measures, viz. the equivalent and compensating variation (EV/CV).
These are hypothetical income adjustments necessary to maintain individual
utilities. While useful for measuring the distribution of change in
individual utility, \textit{adding} Hicksian measures across consumers to
obtain a measure of social welfare change involves judgments that are known
to be conceptually problematic. These include (i) an implicit assumption of
a constant social marginal utility of income, i.e. that an additional dollar
is valued by society identically no matter whether a rich or a poor person
gets it, cf. Blackorby and Donaldson 1988, Dreze 1998, Banks et al 1996,
(ii) ranking alternative interventions by their associated Hicksian
compensations amounts to ranking based on \textit{changes}, as opposed to 
\textit{levels,} of individual satisfaction cf. Slesnick 1998, Chipman and
Moore 1990, and (iii) comparing allocations via the aggregate compensation
principle, as implied by adding EV/CV across consumers, leads to Scitovszky
(1941) reversals, where two distinct allocations can both dominate and be
dominated by each other in terms of social welfare. These problems with
aggregate compensation criteria have led to widespread use of
Bergson-Samuelson aggregate indirect utility for applied welfare analysis in
public finance.\footnote{%
Aggregate indirect utility embodies interpersonal comparisons of preferences
which, as noted by Hammond 1990 \textquotedblleft ... have to be made if
there is to be any satisfactory escape from Arrow's impossibility theorem,
with its implication that individualistic social choice has to be
dictatorial ... or else that is has to restrict itself to solely
recommending Pareto improvements.\textquotedblright\ } The `log-sum' formula
for consumer surplus (cf. Train 2003, Chap 3.5), routinely used for
welfare-analysis in empirical IO, is precisely the average indirect utility
in the parametric multinomial logit (BLP) model. Interestingly, as shown
below, there is a theoretical link in the discrete choice case between
changes in aggregate indirect-utility and Hicksian compensations for removal
of alternatives. Further, based on the aggregate utility, one can define a
micro-founded measure of `welfare-inequality', analogous to the Atkinson
index of income-inequality. Unlike CV/EV, however, the aggregate utility
requires a normalization for empirical content, as is implicitly assumed in
public finance. In the restrictive but popular special case of quasilinear
preferences, under which demand is income-invariant, the change in indirect
utility approach coincides with the Hicksian/Marshallian ones.

Ahmad and Stern 1984, Mayshar 1990, Hendren and Sprung-Keyser 2020 and
Hendren and Finkelstein 2020 investigated social cost-benefit analysis for 
\emph{marginal} interventions, using the concept of `marginal value of
public funds' (MVPF), defined as the ratio of beneficiaries' marginal
willingness-to-pay for a policy-change at the status-quo to its marginal
cost for the government. This approach does not cover non-marginal
interventions and does not clarify how to account for unobserved preference
heterogeneity across individuals targeted by such non-marginal
interventions. Obviously, the larger the intervention, the poorer the
resulting approximation by marginal cost-benefit analysis (see our
application below). It is also not obvious how one would use the status-quo
MVPF for optimal targeting of interventions. Garcia and Heckman 2022 propose
the net social benefit as an alternative to the MVPF for ranking different
programs based on their opportunity costs. Our results facilitate empirical
calculation of all such quantities.

The next section outlines the theory, presents our main result and provides
related discussions, Section 3 presents an empirical illustration. Section 4
concludes.

\section{Theory}

\subsection{Set-up and Main Result}

There is a population of heterogeneous individuals, each facing a choice
between $J+1$ exclusive, indivisible options. Examples include whether to
attend college, whether to adopt a health-product, choice of phone-plan etc.
Let $N$ represent the quantity of numeraire which an individual consumes in
addition to the discrete good. If the individual has income $Y=y$, and faces
a price $P_{j}=p_{j}$ for the $j$th option, then the budget constraint is $%
N+\sum\limits_{j=1}^{J}Q_{j}p_{j}=y$, $\sum_{j=0}^{J}Q_{j}=1$ where $%
Q_{j}\in \left\{ 0,1\right\} $, $j=0,...,J$ represents the discrete choice,
with $0$ denoting the outside option (set $p_{0}=0$). Individuals derive
satisfaction from both the discrete good as well as the numeraire. Upon
buying the $j$th option, an individual derives utility from it and from
numeraire $y-p_{j}$, denoted by $U_{j}\left( y-p_{j},\eta \right) $, where $%
\eta $ denotes unobserved (by us) preference heterogeneity of unspecified
dimension and distribution. Upon not buying any of the $J$ alternatives, she
enjoys utility from her outside option and the full numeraire $y$, given by $%
U_{0}\left( y,\eta \right) $. Observable characteristics of consumers and/or
the alternatives are suppressed here for notational simplicity. We assume
strict non-satiation in the numeraire, i.e. that $U_{j}\left( \cdot ,\eta
\right) $ and $U_{0}\left( \cdot ,\eta \right) $ are strictly increasing in
their first argument for each realization of $\eta $. On each budget set
defined by the price vector $\mathbf{p}\equiv \left( p_{1},...,p_{J}\right) $
and consumer income $y$, there is a structural probability of choosing
option $j$, denoted by $q_{j}\left( \mathbf{p},y\right) $; that is, if each
member of the entire population were offered income $y$ and price $\mathbf{p}
$, then a fraction $q_{j}\left( \mathbf{p},y\right) $ would buy the $j$th
alternative, with $q_{0}\left( \mathbf{p},y\right) $ denoting not buying any
of the $J$ alternatives, i.e.%
\begin{equation}
q_{j}\left( \mathbf{p},y\right) =\int 1\left\{ U_{j}\left( y-p_{j},\eta
\right) >\max_{\substack{ k\in \left\{ 0,1,...,J\right\}  \\ k\neq j}}%
U_{k}\left( y-p_{k},\eta \right) \right\} dF\left( \eta \right) \text{,}
\label{qbar_defn}
\end{equation}%
where $F\left( \cdot \right) $ denotes the \textit{marginal} distribution of 
$\eta $. Note that the above set-up allows for completely general unobserved
heterogeneity and income effects. Finally, the indirect utility function is
given by 
\begin{equation*}
W\left( \mathbf{p},y,\eta \right) =\max \left\{ U_{0}\left( y,\eta \right)
,U_{1}\left( y-p_{1},\eta \right) ,...,U_{J}\left( y-p_{J},\eta \right)
\right\} \text{.}
\end{equation*}%
Note that this function is decreasing in each price and increasing in
income. Therefore, a concave functional of $W\left( \mathbf{p},y,\eta
\right) $ will correspond to assigning larger weights to those with lower
utility and income.

\textbf{Normalization}: Since a monotone transformation of a utility
function represents the same ordinal preferences and leads to the same
choice, we need to normalize one of the alternative-specific utility
functions in order to give empirical content to the indirect utility
function. Toward that end, suppose that for each $\eta $, the function $%
U_{0}\left( \cdot ,\eta \right) $ is strictly increasing (non-satiated) and
continuous in the numeraire and, therefore, invertible. Then $U_{j}\left(
y-p_{j},\eta \right) \geq U_{0}\left( y,\eta \right) $ if and only if $%
U_{0}^{-1}\left( U_{j}\left( y-p_{j},\eta \right) ,\eta \right) \geq y$;
also $U_{0}^{-1}\left( U_{j}\left( y-p_{j},\eta \right) ,\eta \right) \geq
U_{0}^{-1}\left( U_{k}\left( y-p_{k},\eta \right) ,\eta \right) $ if and
only if $U_{j}\left( y-p_{j},\eta \right) \geq U_{k}\left( y-p_{k},\eta
\right) $. Therefore, $\mathcal{U}_{0}\left( y,\eta \right) \equiv y$ and $%
\mathcal{U}_{j}\left( y-p_{j},\eta \right) \equiv U_{0}^{-1}\left(
U_{j}\left( y-p_{j},\eta \right) ,\eta \right) $ is an equivalent
normalization of utilities representing exactly the same individual
preferences as $\left\{ U_{j}\left( y-p_{j},\eta \right) \right\} ,$ $%
j=1,...,J\ $and $U_{0}\left( y,\eta \right) $. This is analogous to the
empirical IO convention that utility from the outside good be normalized to
zero. Arbitrary functional-form specification for utilities (e.g.
CES/CARA/CRRA) in traditional structural modelling assumes much more, in
addition to an implicit normalization. Further, welfare-changes often result
from removing/adding inside alternatives, normalizing utility of the \textit{%
outside} option which remains unaffected by such changes, therefore seems
natural. It also leads to an interpretation of indirect utility as a
compensated income (see Sec 2.3 below). So, from now, we will work under
this normalization.\smallskip 

\begin{theorem}
In the above set-up, assume that $U_{0}\left( \cdot ,\eta \right) $ is
continuous and strictly increasing. Then the marginal distribution of
indirect utility induced by the distribution of $\eta $ at fixed values of $%
\mathbf{p}$ and $y$ is nonparametrically identified from average demand.
\end{theorem}

\begin{proof}
Using the normalization $\mathcal{U}_{j}\left( y-p_{j},\eta \right) \equiv
U_{0}^{-1}\left( U_{j}\left( y-p_{j},\eta \right) ,\eta \right) $ and $%
\mathcal{U}_{0}\left( y,\eta \right) \equiv y$, the indirect utility equals%
\begin{equation}
W\left( \mathbf{p},y,\eta \right) =\max \left\{ y,U_{0}^{-1}\left(
U_{1}\left( y-p_{1},\eta \right) ,\eta \right) ,...,U_{0}^{-1}\left(
U_{J}\left( y-p_{J},\eta \right) ,\eta \right) \right\} \text{.}  \label{1}
\end{equation}%
We wish to compute the (structural) distribution function of $W\left( 
\mathbf{p},y,\eta \right) $ induced by the marginal distribution of $\eta $,
for fixed $\mathbf{p},y$

Now, note that by (\ref{1}), $W\left( \mathbf{p},y,\eta \right) \geq y$ a.s.
Therefore, take $c>y$, and note that%
\begin{eqnarray*}
&&\Pr \left[ \max \left\{ U_{0}^{-1}\left( U_{1}\left( y-p_{1},\eta \right)
,\eta \right) ,...,U_{0}^{-1}\left( U_{J}\left( y-p_{J},\eta \right) ,\eta
\right) \right\} \leq c\right] \\
&=&\Pr \left[ \max \left\{ U_{1}\left( y-p_{1},\eta \right) ,...,U_{J}\left(
y-p_{J},\eta \right) \right\} \leq U_{0}\left( c,\eta \right) \right] \text{%
, since }U_{0}\left( \cdot ,\eta \right) \text{ cont. and strictly}\nearrow
\\
&=&\Pr \left[ \max \left\{ U_{1}\left( c-\left( c-y+p_{1}\right) ,\eta
\right) ,...,U_{J}\left( c-\left( c-y+p_{J}\right) ,\eta \right) \right\}
\leq U_{0}\left( c,\eta \right) \right] \\
&=&q_{0}\left( c-y+p_{1},..c-y+p_{J},c\right) \text{.}
\end{eqnarray*}%
Therefore, the C.D.F.\ of $W\left( \mathbf{p},y,\eta \right) $ generated by
randomness in $\eta $ is given by%
\begin{equation}
\Pr \left[ W\left( \mathbf{p},y,\eta \right) \leq c\right] =\left\{ 
\begin{array}{l}
0\text{ if }c<y \\ 
q_{0}\left( c-y+p_{1},..c-y+p_{J},c\right) \text{ if }c\geq y%
\end{array}%
\right.  \label{4}
\end{equation}
\end{proof}

It is easily verified (see appendix) that $U_{0}\left( \cdot ,\eta \right) $
being strictly increasing implies that the C.D.F. is non-decreasing in $c$.

To interpret (\ref{4}) intuitively, note that for $c\geq y$, (\ref{4}) is
equivalent to 
\begin{equation*}
\Pr \left[ W\left( \mathbf{p},y,\eta \right) >c\right] =1-q_{0}\left(
c-y+p_{1},..c-y+p_{J},c\right) \text{.}
\end{equation*}%
Indeed, if $c\geq y$, the only $\eta $'s who attain a welfare value larger
than $c$ must buy one of $\left\{ 1,...,J\right\} $, since choosing $0$
yields $y\leq c$, which explains the functional form $1-q_{0}(\cdot )$. The
arguments $\left( c-y+p_{1},..c-y+p_{J},c\right) $ arise from the facts that
reaching utility larger than $c$ requires that choosing the maximand $j\in
\left\{ 1,...,J\right\} $ and ending up with numeraire $y-p_{j}$ must yield
higher utility than choosing $0$ at income $c$ which yields utility $c$.

Lastly, note that by definition, $W\left( \mathbf{p},y,\eta \right) $ is
measured in units of money, which will be useful both in cost-benefit
analysis and in comparison with Hicksian compensation.\smallskip

\subsection{Social Welfare Calculations}

Social welfare (Bergson 1938, Atkinson 1970) at price $\mathbf{p}$ for
individuals with income $y$ and unobserved heterogeneity $\eta $ is given by 
$\frac{W\left( \mathbf{p},y,\eta \right) ^{1-\varepsilon }}{1-\varepsilon }$%
, where $0\leq \varepsilon <1$ denotes the planner's inequality aversion
parameter. Therefore, the distribution of social welfare at fixed income $y$
across consumers follows from (\ref{4}). In particular, average (over
unobserved heterogeneity) social welfare (ASW henceforth) at income $y$ is%
\begin{equation}
\mathcal{W}^{\varepsilon }\left( \mathbf{p},y\right) \equiv \int \frac{%
W\left( \mathbf{p},y,\eta \right) ^{1-\varepsilon }}{1-\varepsilon }dF\left(
\eta \right) =E_{\eta }\left\{ \frac{W\left( \mathbf{p},y,\eta \right)
^{1-\varepsilon }}{1-\varepsilon }\right\} \text{,}  \label{3}
\end{equation}%
where $E_{\eta }$ denotes expectation taken with respect to the marginal
distribution of $\eta $. Note that (\ref{3}) takes the form of an exact
counterpart for measuring income inequality using compensated -- instead of
ordinary -- income. For example, one can compute the analogs of the Gini
coefficient or the Atkinson index for welfare inequality based on the
distribution of $W\left( \mathbf{p},y,\eta \right) $ simply by replacing
ordinary income by the compensated income as defined in the LHS of (\ref{5})
below. Calculation of $\mathcal{W}^{\varepsilon }\left( \mathbf{p},y\right) $%
, is facilitated by the observation that%
\begin{equation}
\int_{b}^{\infty }x^{\alpha }f_{X}\left( x\right) dx=b^{\alpha }\left(
1-F_{X}\left( b\right) \right) +\alpha \int_{b}^{\infty }x^{\alpha -1}\left(
1-F_{X}\left( x\right) \right) dx  \label{12}
\end{equation}%
using integration by parts. Therefore, from (\ref{4}), (\ref{3}) and (\ref%
{12}), $\mathcal{W}^{\varepsilon }\left( \mathbf{p},y\right) $ equals%
\begin{eqnarray}
&&\frac{y^{1-\varepsilon }}{1-\varepsilon }\times \Pr \left[ W\left( \mathbf{%
p},y,\eta \right) =y\right] +\int_{y}^{\infty }\frac{c^{1-\varepsilon }}{%
1-\varepsilon }\times f_{W\left( \mathbf{p},y,\eta \right) }\left( c\right)
dc  \notag \\
&=&\frac{y^{1-\varepsilon }}{1-\varepsilon }+\int_{0}^{\infty }\left(
z+y\right) ^{-\varepsilon }\times \left\{ 1-q_{0}\left(
z+p_{1},z+p_{2},...,z+p_{J},z+y\right) \right\} dz\text{.}\footnotemark 
\label{2}
\end{eqnarray}%
\footnotetext{%
For standard parametric CDFs like probit or logit, the integral is bounded
for $0\leq \varepsilon \leq 1$.}For $\varepsilon =0$, i.e. utilitarian
planner preferences, (\ref{2}) reduces to the line integral%
\begin{equation}
y+\int_{0}^{\infty }\left\{ 1-q_{0}\left(
z+p_{1},z+p_{2},...,z+p_{J},z+y\right) \right\} dz  \label{7}
\end{equation}

\textbf{Optimal Targeting Problem}: The optimal subsidy targeting problem
maximizes aggregate welfare subject to a budget constraint on subsidy
spending. Suppose in our multinomial set-up, the planner considers
subsidizing alternative 1. Let $M$ denote the aggregate subsidy budget,
expressed in per capita terms, $F_{Y}\left( \cdot \right) $ the marginal
distribution of income in the population, $\sigma \left( y\right) $ the
amount of subsidy that a household with income $y$ will be entitled to, $%
\mathcal{T}$ denote the set of politically/practically feasible targeting
rules $\sigma \left( \cdot \right) $, and $C\left( y,\sigma \left( y\right)
\right) $ the cost per capita of offering subsidy $\sigma \left( y\right) $
to individuals whose income is $y$; for example, in the multinomial case
with alternative 1 being subsidized, $C\left( y,\sigma \left( y\right)
\right) $ equals $\int \sigma \left( y\right) \times q_{1}\left( \bar{p}%
-\sigma \left( y\right) ,y\right) dF_{Y}\left( y\right) $. Then the optimal
subsidy solves 
\begin{equation}
\arg \max_{\sigma \left( \cdot \right) \in \mathcal{T}}\int \mathcal{W}%
\left( \bar{p}-\sigma \left( y\right) ,\mathbf{p}_{-1},y\right) dF_{Y}\left(
y\right) \text{ s.t. }\int C\left( y,\sigma \left( y\right) \right)
dF_{Y}\left( y\right) =M\text{.}  \label{6}
\end{equation}%
Taxes can be incorporated into the analysis by allowing $\mathcal{T}$ to
contain functions that take on negative values. In particular, a
revenue-neutral welfare maximizing rule that taxes the rich, i.e. $\sigma
\left( \cdot \right) <0$ and subsidizes the poor i.e. $\sigma \left( \cdot
\right) >0$, will solve (\ref{6}) with $M=0$.\smallskip

\textbf{Comparison with Income-transfer}: Price subsidies, as opposed to a
pure income-transfer, entails a deadweight loss due to the distortionary
effect of the price-intervention on behavior. In the context of problem (\ref%
{6}), the aggregate value of this excess-burden can be computed as the
difference between $M$ and the value function of problem (\ref{6}). Indeed,
the worldwide discussion of a universal basic income (cf. Banerjee et al
2019) can be informed by calculating the deadweight loss of various
price-subsidies that the UBI would seek to replace.\smallskip

\textbf{Hicksian Interpretation: }Our measure (\ref{1}) can be interpreted
as the Hicksian CV corresponding to removal of all the inside alternatives.
To see this, consider an initial situation where none of the inside
alternatives $1,...,J$ is available $(p_{1}=p_{2}=...=p_{J}=\infty $,
denoted by the price vector $\mathbf{\infty }_{J}$) and an the eventual
situation when they become available at price vector $\mathbf{p}$. Using the
utility functions $\mathcal{U}_{j}\left( y-p_{j},\eta \right)
=U_{0}^{-1}\left( U_{j}\left( y-p_{j},\eta \right) ,\eta \right) $ and $%
\mathcal{U}_{0}\left( y,\eta \right) =y$, the former indirect utility is $y$
since other options are unavailable, and the latter indirect utility is%
\begin{equation*}
\max \left\{ y,U_{0}^{-1}\left( U_{1}\left( y-p_{1},\eta \right) ,\eta
\right) ,...,U_{0}^{-1}\left( U_{J}\left( y-p_{J},\eta \right) ,\eta \right)
\right\} .
\end{equation*}%
Then the CV $CV\left( y,\mathbf{p,}\infty _{J},\eta \right) $ for this
change solves%
\begin{eqnarray}
&&y+CV\left( y,\mathbf{p,}\infty _{J},\eta \right)  \notag \\
&=&\max \left\{ y,U_{0}^{-1}\left( U_{1}\left( y-p_{1},\eta \right) ,\eta
\right) ,...,U_{0}^{-1}\left( U_{J}\left( y-p_{J},\eta \right) ,\eta \right)
\right\}  \notag \\
&=&W\left( \mathbf{p},y,\eta \right) \text{, by (\ref{1}).}  \label{5}
\end{eqnarray}%
Thus the indirect utility at price $\mathbf{p}$ and income $y$ equals the
compensated income at $y$ that equates individual utility when none of the
alternatives $1,...,J$ was available to the utility when they become
available at price $\mathbf{p}$. It follows from (\ref{5}) that the
difference in individual indirect utility between two prices $\mathbf{p}^{0}$
and $\mathbf{p}^{1}$ equals%
\begin{equation}
W\left( \mathbf{p}^{1},y,\eta \right) -W\left( \mathbf{p}^{0},y,\eta \right)
=CV\left( y,\mathbf{p}^{1},\infty _{J},\eta \right) -CV\left( y,\mathbf{p}%
^{0},\infty _{J},\eta \right) .  \label{10}
\end{equation}%
Note however that%
\begin{equation}
W\left( \mathbf{p}^{1},y,\eta \right) -W\left( \mathbf{p}^{0},y,\eta \right)
\neq -CV\left( y,\mathbf{p}^{0},\mathbf{p}^{1},\eta \right) \text{,}
\label{11}
\end{equation}%
(proved in the Appendix); so asking if $\mathbf{p}^{1}$ is worse than $%
\mathbf{p}^{0}$ on the basis of ASW is not the same as asking if the average
CV for a move from $\mathbf{p}^{0}$ to $\mathbf{p}^{1}$ is positive.
Therefore, comparing two situations on the basis of aggregate Hicksian
compensation is different from comparing them based on the Bergson-Samuelson
ASW criterion.\smallskip

\textbf{Comparing ASW with CV}: The CV answers the question: what change in
income $y$ would have resulted in the same change of utility as a given
change of prices $p$, relative to a baseline level of $(y;p)$; whereas the
average indirect utility defined here answers: what income would yield the
same level of utility as a given level $(y;p)$, assuming that individuals
are prohibited from purchasing the discrete options under consideration, but
are free in all other choices. Unlike aggregate CV, the ASW criterion does
not assume that social marginal utility of income is constant across income,
and does not suffer from conceptual ambiguities like Scitovszky reversal of
social preferences (cf. Mas-Colell et al 1995 page 830-31), as illustrated
in Figure \ref{fig:ASWvsCV}.

\begin{figure}[tbp]
\begin{center}
\includegraphics[
height=2.1909in,
width=2.2872in
]{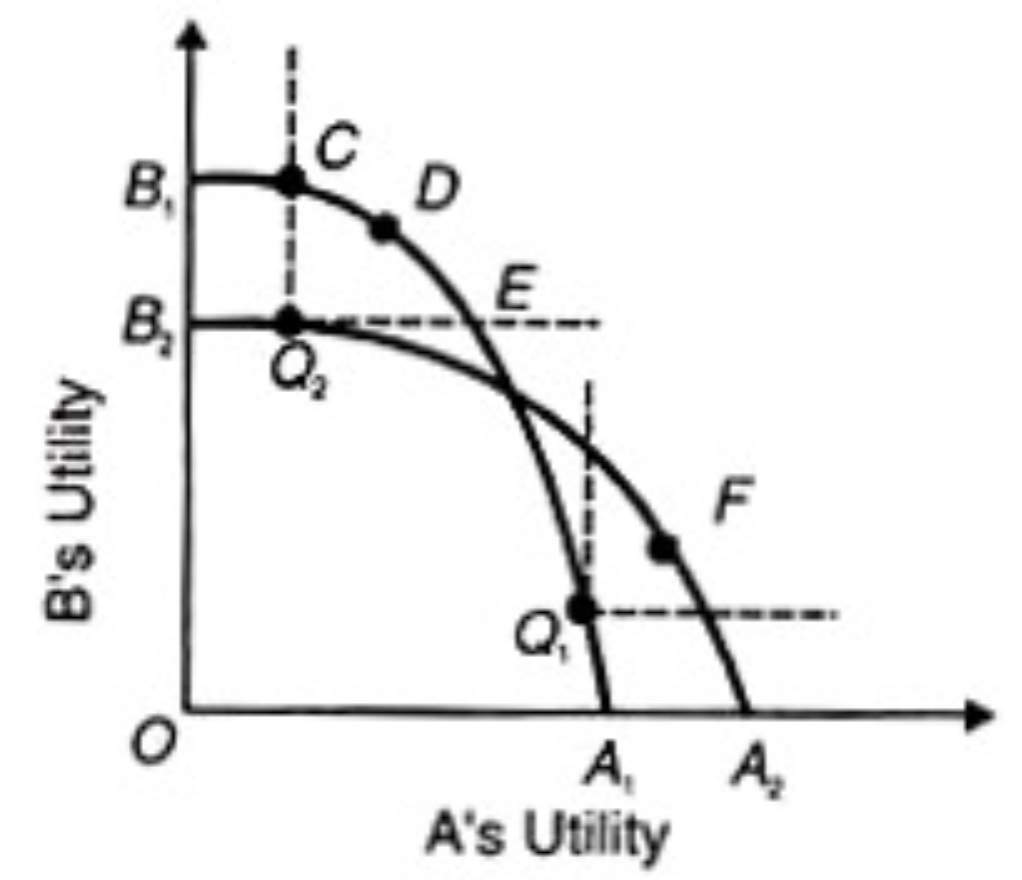}
\end{center}
\caption{Comparison of ASW with CV: An illustration.}
\label{fig:ASWvsCV}
\end{figure}


Figure \ref{fig:ASWvsCV} shows two allocations $Q_{1}$ and $Q_{2}$ with the
utility possibility frontiers $A_{1}Q_{1}DB_{1}$ and $A_{2}FQ_{2}B_{2}$
through them intersecting. Each frontier represents the utility combinations
attainable via redistribution between individuals A and B, starting from any
point on it. Then the allocation $D$ Pareto dominates $Q_{2}$ and can be
attained from $Q_{1}$ via redistribution. Therefore $Q_{1}$ is \textit{%
superior} to $Q_{2}$ via the aggregate compensation principle. At the same
time, the allocation $F$ which can be attained via redistribution from $%
Q_{2} $ is Pareto superior to $Q_{1}$, implying that $Q_{1}$ is \textit{%
inferior} to $Q_{2}$ via the compensation principle; so aggregate CV is
again negative, thus leading to an ambiguity. These conceptual shortcomings
of the Hicksian approach have led instead to widespread use of the
Bergson-Samuelson ASW criterion in applied research in public finance. In
empirical IO, the widely used log-sum measure of consumer-welfare is
precisely the ASW in the multinomial logit model, cf. Train 2003, Sec
3.5.\smallskip

\textbf{Quasilinear Utilities}: If $U_{j}\left( y-p_{j},\eta \right)
=h_{j}\left( \eta \right) +y-p_{j}$, and $U_{0}\left( y,\eta \right)
=h_{0}\left( \eta \right) +y$, i.e. utility is quasilinear in the numeraire,
then it can be shown (see appendix for derivation) that%
\begin{equation}
W\left( \mathbf{p}^{1},y,\eta \right) -W\left( \mathbf{p}^{0},y,\eta \right)
=-CV\left( y,\mathbf{p}^{0},\mathbf{p}^{1},\eta \right) \text{,}
\end{equation}%
and $W\left( \mathbf{p},y,\eta \right) $ equals%
\begin{equation}
\left\{ y+\max \left\{ 0,h_{1}\left( \eta \right) -p_{1}-h_{0}\left( \eta
\right) ,...,h_{J}\left( \eta \right) -p_{J}-h_{0}\left( \eta \right)
\right\} \right\} dF\left( \eta \right) \text{,}  \label{18a}
\end{equation}%
so that the social marginal utility of income $\frac{\partial }{\partial y}%
\int W\left( \mathbf{p},y,\eta \right) dF\left( \eta \right) $ equals 1,
which does not depend on $y$; i.e. society is indifferent between giving a
dollar to a rich versus a poor individual. Thus aggregate CV gives a
legitimate measure of social welfare when the social marginal utility of
income is constant (cf. Blackorby and Donaldson 1988).

\subsection{Binary Choice}

Our application is a binary choice setting with $J=1$, where a subsidy on
alternative 1 changes its price from $\bar{p}$ to $\bar{p}-\sigma $. In this
case, the subsidy-induced \textbf{change in average social welfare} at
income $y$ for a generic $\varepsilon \geq 0$ is given by%
\begin{equation}
\Delta \left( \bar{p},\sigma ,y;\varepsilon \right) \equiv \int_{0}^{\infty
}\left( z+y\right) ^{-\varepsilon }\times \left[ q_{1}\left( \bar{p}-\sigma
+z,y+z\right) -q_{1}\left( \bar{p}+z,y+z\right) \right] dz\text{;}
\label{8d}
\end{equation}%
under $\varepsilon =0$, we have that%
\begin{equation}
\Delta \left( \bar{p},\sigma ,y;0\right) \equiv \int_{0}^{\infty }\left[
q_{1}\left( \bar{p}-\sigma +z,y+z\right) -q_{1}\left( \bar{p}+z,y+z\right) %
\right] dz\text{;}  \label{8a}
\end{equation}%
the \textbf{average CV} at income $y$ (cf. Bhattacharya 2015, eqn. 10) equals%
\begin{equation}
S\left( \bar{p},\sigma ,y\right) \equiv \int_{\bar{p}-\sigma }^{\bar{p}%
}q_{1}\left( p,y+p-\bar{p}+\sigma \right) dp\overset{\text{subs }z=p-\bar{p}%
+\sigma }{=}\int_{0}^{\sigma }q_{1}\left( \bar{p}-\sigma +z,y+z\right) dz%
\text{.}  \label{8b}
\end{equation}%
Furthermore,%
\begin{equation}
\int_{0}^{\infty }q_{1}\left( \bar{p}+z,y+z\right) dz\overset{\text{sub }%
t=z+\sigma }{=}\int_{\sigma }^{\infty }q_{1}\left( \bar{p}+t-\sigma
,y+t-\sigma \right) dt  \label{8f}
\end{equation}%
and therefore, from (\ref{8a}), (\ref{8b}) and (\ref{8f}) we have that%
\begin{eqnarray}
&&\Delta \left( \bar{p},\sigma ,y;0\right) -S\left( \bar{p},\sigma ,y\right)
\notag \\
&=&\int_{\sigma }^{\infty }\left[ q_{1}\left( \bar{p}-\sigma +z,y+z\right)
-q_{1}\left( \bar{p}-\sigma +z,y-\sigma +z\right) \right] dz  \label{8g}
\end{eqnarray}%
Now, the integrand in (\ref{8g}) is strictly positive (negative) for all $z$
if option 1 is normal (resp. inferior). Therefore, the only way $\Delta
\left( \bar{p},\sigma ,y;0\right) =S\left( \bar{p},\sigma ,y\right) $ is
that $q_{1}\left( p,y\right) $ does not depend on $y$, which implies
utilities are quasilinear, and therefore by (\ref{18a}), the social marginal
utility of income equals 1.

The \textbf{average treatment effect}, the quantity most commonly used in
program evaluation and the treatment choice literature, equals%
\begin{equation}
T\left( \bar{p},\sigma ,y\right) \equiv q_{1}\left( \bar{p}-\sigma ,y\right)
-q_{1}\left( \bar{p},y\right) \text{,}  \label{8c}
\end{equation}%
which is simply the integrand of (\ref{8a}) evaluated at the lower limit of
the integral. Since this is measured as \textit{quantity} of demand, a
direct comparison with average or marginal subsidy cost is not possible. In
contrast, the quantities $\Delta \left( \cdot ,\cdot ,\cdot \right) $ or $%
S\left( \cdot ,\cdot ,\cdot \right) $ provide theoretically justified
monetary values of the choice, based on the choice-makers' own preference.

\textbf{Deadweight Loss (DWL)}: The average cost of the subsidy equals $%
\sigma \times q_{1}\left( \bar{p}-\sigma ,y\right) $ in every case.
Therefore, the DWL of the subsidy under $\varepsilon =0$ is given by%
\begin{eqnarray}
&&\sigma \times q_{1}\left( \bar{p}-\sigma ,y\right) -\Delta \left( \bar{p}%
,\sigma ,y;0\right)  \notag \\
&=&\int_{\bar{p}-\sigma }^{\bar{p}}\left[ q_{1}\left( \bar{p}-\sigma
,y\right) -q_{1}\left( t,y+t-\bar{p}+\sigma \right) \right] dt  \notag \\
&&+\int_{\bar{p}}^{\infty }\left[ q_{1}\left( t,y+t-\bar{p}\right)
-q_{1}\left( t,y+t-\bar{p}+\sigma \right) \right] dt\text{.}  \label{22}
\end{eqnarray}%
Note that the first term in (\ref{22}) is positive because%
\begin{eqnarray}
&&q_{1}\left( t,y+t-\bar{p}+\sigma \right) -q_{1}\left( \bar{p}-\sigma
,y\right)  \notag \\
&=&\Pr \left[ U_{1}\left( y-\bar{p}+\sigma ,\eta \right) \geq U_{0}\left(
y+t-\bar{p}+\sigma ,\eta \right) \right] -\Pr \left[ U_{1}\left( y-\bar{p}%
+\sigma ,\eta \right) \geq U_{0}\left( y,\eta \right) \right]  \notag \\
&\leq &0\text{, for }t\geq \bar{p}-\sigma \text{ since }U_{0}\left( \cdot
,\eta \right) \text{ is strictly increasing.}  \label{23}
\end{eqnarray}%
The second term will be negative if the good is normal, and the DWL may be 
\textit{negative} if the income effect is strongly positive. This is in
contrast to the deadweight loss based on the CV which must necessarily be
non-negative.

Finally, Hendren-Finkelstein's MVPF at the status-quo ($p=\bar{p}$, $\sigma
=0$) equals%
\begin{equation}
\left. \frac{\frac{\partial }{\partial \sigma }\Delta \left( \bar{p},\sigma
,y;0\right) }{\frac{\partial }{\partial \sigma }\left\{ \sigma \times
q_{1}\left( \bar{p}-\sigma ,y\right) \right\} }\right\vert _{\sigma =0}%
\overset{\text{by (\ref{8a})}}{=}\frac{-\int_{0}^{\infty }\frac{\partial
q_{1}\left( \bar{p}+z,y+z\right) }{\partial p}dz}{q_{1}\left( \bar{p}%
,y\right) }  \label{MVPFratio}
\end{equation}

\subsection{Treatment targeting}

The constrained, optimal subsidy allocation problem takes the form%
\begin{equation}
\max_{\sigma \left( \cdot \right) \in \mathcal{T}}\int B\left( \bar{p}%
,\sigma \left( y\right) ,y\right) dF_{Y}\left( y\right) \text{ s.t. }\int
\sigma \left( y\right) \times q_{1}\left( \bar{p}-\sigma \left( y\right)
,y\right) dF_{Y}\left( y\right) \leq M\text{,}  \label{8e}
\end{equation}%
where $M$ denotes the planner's budget constraint, $\mathcal{T}$ denotes the
set of politically/practically feasible targeting rules, $F_{Y}\left( \cdot
\right) $ is the marginal distribution of income in the population, and $%
B\left( \cdot ,\cdot ,\cdot \right) $ is one of $\Delta \left( \cdot ,\cdot
,\cdot \right) $, $S\left( \cdot ,\cdot ,\cdot \right) $ or $T\left( \cdot
,\cdot ,\cdot \right) $, defined in (\ref{8a})-(\ref{8c}).\smallskip

\textbf{Parameter Uncertainty}: Note that (\ref{8e}) seeks to maximize
welfare of the individuals we observe. If instead, we treat our data as a
random sample from a population, and wish to maximize welfare for that
population, then we would need to take parameter uncertainty into account.
This can be done by defining a loss function%
\begin{eqnarray}
L\left( \sigma \left( \cdot \right) ,\theta ,c\right) &=&-\int B\left( \bar{p%
},\sigma \left( y\right) ,y,\theta \right) dF\left( y,\theta _{1}\right) 
\notag \\
&&+c\left[ M-\int \sigma \left( y\right) \times q_{1}\left( \bar{p}-\sigma
\left( y\right) ,y,\theta _{2}\right) dF\left( y,\theta _{1}\right) \right]
^{2}\text{,}  \label{lossfunctiongeneral}
\end{eqnarray}%
where $c$ denotes the penalty incurred by the planner from violating the
budget constraint, and $\theta _{1}$, $\theta _{2}$ denote the parameters
determining the marginal distribution of income and the demand function e.g.
logit coefficients, respectively. Then define the optimal choice of $\sigma
\left( \cdot \right) $ under a Bayesian criterion by solving%
\begin{equation}
\min_{\sigma \left( \cdot \right) }\int L\left( \sigma \left( \cdot \right)
,\theta ,c\right) dP_{post}\left( \theta |data\right)
\label{lossfunctiongeneral2}
\end{equation}%
where $P_{post}\left( \theta |data\right) $ refers to the posterior
distribution of $\theta $ given the data. For computational simplicity, one
can use the bootstrap distribution of $\theta $ to approximate the posterior
corresponding to a flat prior (cf. Hastie et al 2009).

\subsection{Identification and Estimation}

Theorem 1 expresses the distribution of indirect utility in terms of the
structural choice probability defined in (\ref{qbar_defn}). Learning the
entire distribution of $W\left( \mathbf{p},y,\eta \right) $ at fixed $%
\mathbf{p},y$ would require one to estimate $q_{0}\left(
c-y+p_{1},..c-y+p_{J},c\right) $ for all values of $c$. In any finite
dataset, of course there will be limited variation of prices and income. So
one can use a flexible parametric model such as random coefficients to
estimate $q_{0}\left( c-y+p_{1},..c-y+p_{J},c\right) $ as is popular in
empirical IO; shape restrictions on the choice probability functions (cf.
Bhattacharya, 2021) can be imposed by restricting the support of the random
coefficients. Any such parametric approximation would obviously impose
additional restrictions on preference that are not required for Theorem 1 to
hold.\footnote{%
In particular, in a binary setting, a probit functional form with constant
coefficients implicitly assumes additive scalar unobserved preference
heterogeneity which implies rank invariance across consumers (cf.
Bhattacharya 2021, page 463).} Alternatively, one can remain nonparametric
and work with bounds. In particular,%
\begin{eqnarray*}
&&q_{0}\left( c-y+p_{1},..c-y+p_{J},c\right) \\
&=&\Pr \left[ U_{0}\left( c,\eta \right) \geq \max \left\{ U_{1}\left(
y-p_{1},\eta \right) ...U_{J}\left( y-p_{J},\eta \right) \right\} \right]
\end{eqnarray*}%
and $U_{j}\left( \cdot ,\eta \right) $ being strictly increasing for each $j$%
, yields nonparametric bounds on $q_{0}\left( c-y+p_{1},..c-y+p_{J},c\right) 
$. Specifically, let $S=\left\{ \mathbf{r},z\right\} $ with $\mathbf{r=}%
\left( r_{1},...,r_{J}\right) $ be the set of price-income combinations
observed in sample. Then lower and upper bounds on $q_{0}\left(
c-y+p_{1},..c-y+p_{J},c\right) $ are given by%
\begin{eqnarray*}
LB\left( c;\mathbf{p},y\right) &=&\max \left\{ q_{0}\left( \mathbf{r}%
,z\right) :\left( \mathbf{r},z\right) \in S,\text{ }z-r_{j}\geq y-p_{j},%
\text{ }z\leq c\right\} \\
UB\left( c;\mathbf{p},y\right) &=&\min \left\{ q_{0}\left( \mathbf{r}%
,z\right) :\left( \mathbf{r},z\right) \in S,\text{ }z-r_{j}\leq y-p_{j},%
\text{ }z\geq c\right\} \text{.}
\end{eqnarray*}%
Arguments presented in Bhattacharya 2021, Proposition 1 imply that these
bounds are sharp.

Finally, if price and/or income are endogenous to individual preference,
i.e. independence between utilities and budget set does not hold, then
consistent estimation of $q_{0}$ would require the use of control
function-type methods cf. Rivers-Vuong 1988, Newey 1987, Blundell-Powell
2004. We apply Newey's approach in our empirical illustration below.

\subsection{Ordered Discrete Choice and the Continuous Case}

A result analogous to Theorem 1 does not hold for consumption of continuous
goods such as gasoline (cf. Poterba, 1991) and food (Kochar 2005). To see
why, consider the situation of ordered choice with 0 denoting the outside
good and 1, 2 with unit price $p$ denoting the two inside good (e.g. no
apple, 1 apple and 2 apples). Let the utilities be $U_{0}\left( y,\eta
\right) $, $U_{1}\left( y-p,\eta \right) $, $U_{2}\left( y-2p,\eta \right) $%
. As above, normalize%
\begin{equation*}
W\left( p,y,\eta \right) \equiv \max \left\{ y,U_{0}^{-1}\left( U_{1}\left(
y-p,\eta \right) ,\eta \right) ,U_{0}^{-1}\left( U_{2}\left( y-2p,\eta
\right) ,\eta \right) \right\}
\end{equation*}

Now, for $c\geq y$, we have that%
\begin{eqnarray}
&&\Pr \left[ \max \left\{ y,U_{0}^{-1}\left( U_{1}\left( y-p,\eta \right)
,\eta \right) ,U_{0}^{-1}\left( U_{2}\left( y-2p,\eta \right) ,\eta \right)
\right\} \leq c\right]  \notag \\
&=&\Pr \left[ \max \left\{ U_{1}\left( y-p,\eta \right) ,U_{2}\left(
y-2p,\eta \right) \right\} \leq U_{0}\left( c,\eta \right) \right]  \notag \\
&=&\Pr \left[ \max \left\{ U_{1}\left( c-\left( c-y+p\right) ,\eta \right)
,U_{1}\left( c-\left( c-y+2p\right) ,\eta \right) \right\} \leq U_{0}\left(
c,\eta \right) \right]  \notag \\
&=&q_{0}\left( c-y+p,c-y+2p,c\right) \text{.}  \label{17}
\end{eqnarray}

But $q_{0}\left( c-y+p,c-y+2p,c\right) $ cannot be estimated, no matter how
much $p$ and $y$ vary, because the data can only identify demand when the
price of 2 units is twice the price of 1 unit; but $c-y+2p\neq 2\left(
c-y+p\right) $. The continuous case can be thought of as the limiting case
of ordered choice, e.g. one has to pay twice as much for 2 gallons of
gasoline as for 1 gallon, and by the same logic, the welfare distribution
for this case cannot be point-identified. One can however obtain bounds on (%
\ref{17}) via%
\begin{eqnarray*}
&&L\left( c;p,y\right) \\
&\equiv &\max \left\{ q_{0}\left( \tilde{p},2\tilde{p},\tilde{y}\right) :%
\tilde{y}-\tilde{p}\geq y-p,\tilde{y}-2\tilde{p}\geq y-2p,\tilde{y}\leq
c\right\} \\
&\leq &\Pr \left[ \max \left\{ U_{1}\left( y-p,\eta \right) ,U_{2}\left(
y-2p,\eta \right) \right\} \leq U_{0}\left( c,\eta \right) \right] \\
&\equiv &q_{0}\left( c-y+p,c-y+2p,c\right) \\
&\leq &\min \left\{ q_{0}\left( \tilde{p},2\tilde{p},\tilde{y}\right) :%
\tilde{y}-\tilde{p}\leq y-p,\tilde{y}-2\tilde{p}\leq y-2p,\tilde{y}\geq
c\right\} \\
&\equiv &H\left( c;p,y\right)
\end{eqnarray*}%
and $L\left( c;p,y\right) $ and $H\left( c;p,y\right) $ are both potentially
identifiable because they represent demand in situations where price of
option 2 is twice the price of option 1. Note that $L\left( \cdot
;p,y\right) $ and $H\left( \cdot ;p,y\right) $ satisfy all properties of
CDF's.

\section{Empirical Illustration: Private Tuition in India}

\label{sec:app1}

Private, remedial tuition for children outside schools is ubiquitous in
South Asia. Most of this is provided on a for-profit basis by
school-teachers themselves. This creates perverse incentives for them to
reduce their efforts inside the school classroom, cf. Jayachandran 2014.
Thus children of richer households, who can afford the additional
tuition-fees, benefit from the educational support outside school, whereas
those from poorer households suffer the adverse consequences of
lower-quality classroom-teaching. One possible way to address this problem
is to tax private tuition for richer households and use the tax proceedings
to subsidize poorer children. We investigate, empirically, the impact of
this hypothetical policy intervention on social welfare, using the methods
developed above.

\begin{table}[tbp]
    \caption{Application data summaries and welfare calculations.}
    \label{table:PTvalues1}
    \begin{center}
    {Summary statistics.} 
    \end{center}
    \par
\begin{tabular}{lccccc}
\hline
Variable & Mean & St. dev. & Median & Min & Max \\ \hline
choice & 0.4067 & 0.4912 & 0 & 0 & 1 \\ 
price & 3291.6 & 3689 & 2000 & 0 & 20000 \\ 
monthly income (per household member) & 9579.5 & 6379.2 & 8000 & 100 & 86700
\\ 
household size & 5.8799 & 2.4896 & 5 & 1 & 32 \\ 
age of child & 12.2658 & 3.5747 & 12 & 5 & 18 \\ 
male & 0.5395 & 0.4984 & 1 & 0 & 1 \\ \hline
\end{tabular}%
\newline
\par
\bigskip 
\par 
\vskip 0.2in
    \begin{center}
    {Welfare estimates.} 
    \end{center}

\vskip 0.2in
\begin{tabular}{lcccc}
\hline
Welfare quantity & At median income & Averaged over income &  &  \\ \hline
average compensating variation & 1499.5 & 1501.8 &  &  \\ 
\quad & (30.056) & (31.531) &  &  \\ 
average treatment effect & 0.0561 & 0.0561 &  &  \\ 
\quad & (0.0134) & (0.0135) &  &  \\ 
change in average social welfare, $\varepsilon=0$ & 1523.4 & 1520.9 &  &  \\ 
\quad & (281.822) & (279.252) &  &  \\ 
change in average social welfare, $\varepsilon=0.5$ & 10.51 & 10.43 &  &  \\ 
\quad & (1.681) & (1.672) &  &  \\ 
change in average social welfare, $\varepsilon=1$ & 0.0775 & 0.0787 &  &  \\ 
\quad & (0.0132) & (0.0136) &  &  \\ \hline
\end{tabular}%
\vskip 0.2in

{\small Welfare quantities in (\ref{8a})-(\ref{8c}) calculated at $\bar{p}%
=4200$ (75th percentile of the price distribution), $\bar{p}-\sigma =900$
(25th percentile). Covariate values are at their median values. Bootstrap
standard errors in parentheses are based on 400 replications.}
\end{table}

We use micro-data from India's National Sample Survey 71st round, conducted
in January-June 2014. The key variables and summary statistics are reported
in Table \ref{table:PTvalues1}. Our sample size is 51092. There are two
important data issues here. Firstly, if a household does not purchase
private tuition, we do not observe their potential spending had they bought
it. This is a well-known empirical issue in discrete choice applications; we
address it by using the average price of those opting for private tuition in
the village/block of the reference household to impute that price. An
intuitive justification is that households are likely to base their decision
on information they gather from acquaintances, and tuition-rates are
unlikely to vary much within a neighborhood. A second issue, given that the
data are non-experimental, is that prices are likely to be correlated with
unobserved tuition-quality. We address this using `Hausman-instruments'
which are average price in other villages/blocks in the same strata
(sampling areas larger than blocks but smaller than districts). The
first-stage F-statistic has a p-value of $10^{-5}$.

We model demand for private tuition as 
\begin{equation*}
q_{1}(p,y)=\Phi \left( \beta _{p}p+\sum_{m=1}^{M}\beta _{y,m}\mathcal{R}%
_{m,M+q}(y;q)+x^{\prime }\beta \right) ,
\end{equation*}%
where $p$ denotes price, $\mathcal{R}_{m,M+q}(y;q)$ are base B-splines of
degree $q$ in monthly income $y$ and the covariate vector $x$ representing
household size, the child's age and sex. We use Newey's (1987) two-step
estimation approach that assumes joint normality of $u$ in the model for the
latent variable 
\begin{equation*}
q_{1}^{\ast }=\beta _{p}p+\sum_{m=1}^{M}\beta _{y,m}\mathcal{R}%
_{m,M+q}(y;q)+x^{\prime }\beta +u
\end{equation*}%
and the error in the reduced-form equation for price. To impose shape
constraints, in the second step, we require 
\begin{align*}
& \beta _{p}\leq 0, \\
& \beta _{p}+\sum_{m=2}^{M}\dfrac{\beta _{y,m}-\beta _{y,m-1}}{z_{m+q}-z_{m}}%
\mathcal{R}_{m-1,M+q-1}(y;q-1)\leq 0
\end{align*}%
with the second inequality imposed on a finite grid of income values. In the
second inequality, $z_{m}$ denote knots on the support of income used in the
construction of B-splines. The first (second) inequality guarantees that $q_{1}(p,y)$
is decreasing in price (in the direction $(1,1)$, i.e. $\frac{\partial }{\partial p}q_{1}(p,y)+\frac{\partial }{%
\partial y}q_{1}(p,y)\leq 0$ (cf. Bhattacharya 2021)). 

The left of Figure \ref{fig:ApptPT_ACVASW_minus_cost} plots demand as a
function of price for fixed income, and as a function of income for fixed
price for a household with a representative set of characteristics. The
inverted U-shape of the income graph agrees with widespread anecdotal
evidence that academic success is primarily a middle-class aspiration in
India, cf. Varma 2007.

The middle of Figure \ref{fig:ApptPT_ACVASW_minus_cost} shows the ACV and
change in ASW, net of average cost, at median income and price over a range
of subsides. We approximated the change in ASW at a given income value $y$
by calculating integrals in the definition of CASW from 0 to $y_{max}-y$,
where $y_{max}$ denotes the largest observed income. This approximation is
quite accurate as the values of integrands around $y_{max}-y$ are
decreasing, taking values less than 0.0003. A negative subsidy is a tax, and
the corresponding net-benefit equals the tax-revenue less utility-loss. We
consider taxes up to 20\% of median price and subsidy up to $Med(p)-\min (p)$%
. In the same figure, we plot the net-benefit approximation by the MVPF,
i.e. $\sigma \times (numerator-demonominator)$ of (\ref{MVPFratio}). This
curve is a straight line through the origin, showing the declining accuracy
of first-order approximation as $\sigma $ rises.

The rightmost panel of Figure \ref{fig:ApptPT_ACVASW_minus_cost} shows the
difference between ACV and change in ASW, when income-effects are/aren't
allowed.

\begin{sidewaysfigure}[tbp]
\begin{center}
\vspace{0in} 
\begin{minipage}{\linewidth}
\begin{minipage}{0.33\linewidth}
  \includegraphics[width=0.95\linewidth]{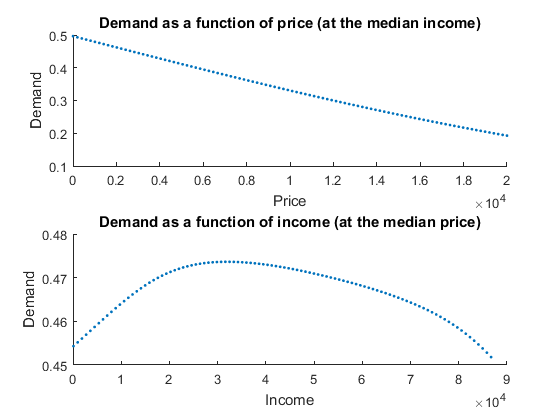}
\end{minipage}
\begin{minipage}{0.33\linewidth}
\includegraphics[width=0.95\linewidth]{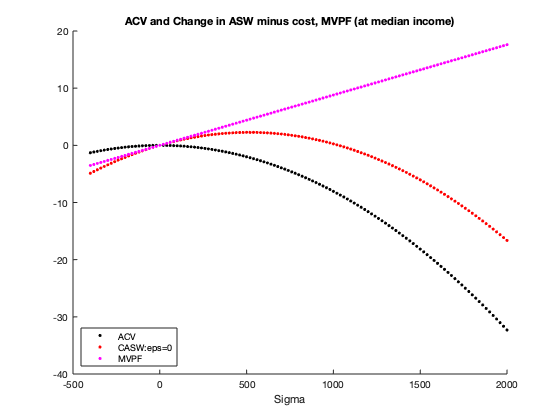}
\end{minipage}
\begin{minipage}{0.33\linewidth}
\includegraphics[width=0.95\linewidth]{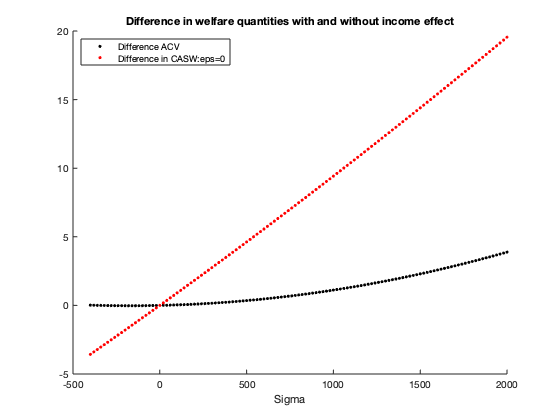}
\end{minipage}
\end{minipage}
\end{center}
\caption{{\protect\small Left: Illustration of demand estimation. Middle: ACV, change in
ASW ($\protect\varepsilon =0$) net of average cost, and the first-order effect based on MVPF ($\bar{p}$ and $y$ are at their median levels). Covariate
values are at their median levels. Right: Difference in ACV, change in
ASW ($\protect\varepsilon =0$) with and without income effect.}}
\label{fig:ApptPT_ACVASW_minus_cost}
\end{sidewaysfigure}

The income-effect is strong; consequently, the  ACV and change in ASW curves differ substantially. Secondly, while deadweight loss
for ACV is necessarily positive, that for the ASW is actually \textit{%
negative} (benefit exceeds cost) over a range of income, which empirically
illustrates our discussion around eqn. (\ref{23}) above.

Table \ref{table:PTvalues1} reports changes in ASW for $\epsilon =0,0.5,1$,
ACV and ATE calculated at $\bar{p}$ equal to the 75th percentile of the
price distribution, $\bar{p}-\sigma $ equal to the 25th, covariates set
equal to their median values, and $y$ set equal to median income $Med(y)$.
We also include bootstrap standard errors.

A natural treatment-assignment problem in this case is to optimally
subsidize the poor by taxing the rich in a budget-neutral way. The formal
problem, analogous to (\ref{8e}) is%
\begin{equation}
\max_{\sigma \left( \cdot \right) \in \mathcal{T}}\int B\left( \bar{p}%
,\sigma \left( y\right) ,y\right) dF_{Y}\left( y\right) \text{ s.t. }\int
\sigma \left( y\right) \times q_{1}\left( \bar{p}-\sigma \left( y\right)
,y\right) dF_{Y}\left( y\right) =0\text{,}  \label{24}
\end{equation}%
where $\mathcal{T}$ is now the space of spline functions which can take both
positive and negative values. The optimal allocation where $B\left( \cdot
,\cdot ,\cdot \right) $ corresponds to average treatment-effect, change in
ASW with $\varepsilon =0$ and average CV are shown in Figure \ref{fig:App2allocations}.

\begin{sidewaysfigure}[tbp]
\begin{center}
\vspace{0in} 
\begin{minipage}{\linewidth}
\begin{minipage}{0.5\linewidth}
  \includegraphics[width=0.95\linewidth]{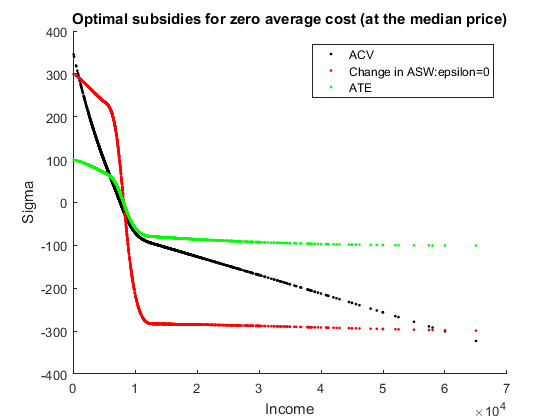}
\end{minipage}
\begin{minipage}{0.5\linewidth}
\includegraphics[width=0.95\linewidth]{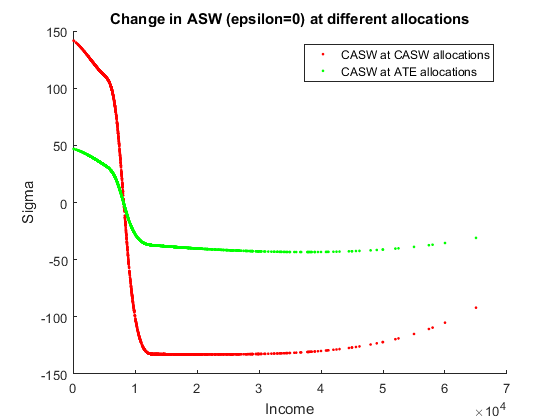}
\end{minipage}
\end{minipage}
\end{center}
\caption{{\protect\small Left: Optimal subsidies for three different welfare
criteria with the budget constraint giving the zero average cost. Right: Change in ASW ($\protect\epsilon=0$) calculated at
the CASW ($\protect\epsilon=0$) optimal allocation path and at the ATE
optimal allocation path.}}
\label{fig:App2allocations}
\end{sidewaysfigure}

Three features stand out. First, the allocation maximizing the ACV
differs from the one that maximizes the CASW at $\epsilon =0$; this results
from the large income-effect. Second, all three curves are downward sloping,
because price-sensitivity of demand declines monotonically with income and
the overall price-effect is much stronger than the income effect (see
appendix Section \ref{sec:Appendix_treatmentshape} for details). Third, the
ATE-maximizing allocation leads to the minimum disparity in subsidy/tax
rates across the rich and poor, whereas maximizing the CASW leads to the
highest disparity where the poor receive the highest subsidy and the rich
face the highest tax. The optimal ACV curve lies in between.

The right panel in Figure \ref{fig:App2allocations} shows 
CASW when we use the ATE maximizing allocations compared to the CASW
at the CASW-optimal allocations. Evidently, using the ATE-optimal allocation leads to much smaller redistribution of
welfare from the rich to the poor, relative to the CASW-optimal allocation.
Most of this operates at the intensive margin since the two graphs cross
zero close to each other. That is, almost the same set of individuals sees
an increase and decrease in their welfare in the two optimal allocations,
but the extent of welfare-gain is much higher for the poor when using the
CASW-optimal allocations. The figure looks somewhat similar to the optimal
subsidy graph because the price-effect on demand -- and, hence, aggregate
welfare -- is much stronger than the
income-effect. At high incomes, demand becomes less price-sensitive, which
explains why the similarity declines there.

For the alternative version that incorporates sampling uncertainty, the
loss-function analogous to (\ref{lossfunctiongeneral}) is%
\begin{eqnarray}
L\left( \sigma \left( \cdot \right) ,\theta ,c\right) &=&-\int B\left( \bar{p%
},\sigma \left( y\right) ,y,\theta _{1}\right) dF\left( y,\theta _{2}\right)
\notag \\
&&+c\left[ \int \sigma \left( y\right) \times q_{1}\left( \bar{p}-\sigma
\left( y\right) ,y,\theta _{1}\right) dF\left( y,\theta _{2}\right) \right]
^{2}\text{,}
\end{eqnarray}%
The optimal subsidy would solve (\ref{lossfunctiongeneral2}) with this loss
function (results not reported for brevity).

\section{Conclusion}

We show how to incorporate social welfare within traditional econometric
program-evaluation and statistical treatment-assignment problems. Our main
result pertains to the practically important setting of multinomial choice.
The key insight is that the marginal distribution of suitably normalized
individual indirect utility can be expressed as a closed-form functional of
choice probabilities without functional-form assumptions on unobserved
preference heterogeneity and income-effects. This leads to expressions for
average weighted social welfare with weights reflecting planners'
distributional preferences and the optimal targeting of interventions that
maximize aggregate utility under fixed budget. We discuss practical issues
of identification and estimation, connections with and advantages relative
to aggregate Hicksian welfare-measures, and potential extension to ordered
and continuous choice. We illustrate our results using the example of
private tuition demand in India where optimal, income-contingent targeting
of subsidies/taxes leads to very different paths depending on whether
average uptake or average social welfare is being maximized. The main source
of this difference is how price-elasticity of demand varies with income.

\pagebreak

\begin{center}
\textbf{References}
\end{center}

\begin{enumerate}
\item Ahmad, E. and Stern, N., 1984. The theory of reform and Indian
indirect taxes. Journal of Public economics, 25(3), pp.259-298.

\item Amemiya, T. 1978. The estimation of a simultaneous equation
generalized probit model. Econometrica 46, pp.1193-1205.

\item Atkinson, A.B., 1970. On the measurement of inequality. Journal of
economic theory, 2(3), 244-263.

\item Azam, M., 2016. Private tutoring: evidence from India. Review of
Development Economics, 20(4), pp.739-761.

\item Banerjee, A., Niehaus, P. and Suri, T., 2019. Universal basic income
in the developing world. Annual Review of Economics, 11, pp.959-983.

\item Banks, J., Blundell, R. and Lewbel, A., 1996. Tax reform and welfare
measurement: do we need demand system estimation?. The Economic Journal,
106(438), 1227-1241.

\item Bergson, A., 1938. Reformulation of Certain Aspects of Welfare
Economics. Quarterly Journal of Economics, 52.

\item Bhattacharya, D., 2015. Nonparametric welfare analysis for discrete
choice. Econometrica, 83(2), pp.617-649.

\item Bhattacharya, D., 2018. Empirical welfare analysis for discrete
choice: Some general results. Quantitative Economics, 9(2), pp.571-615.

\item Bhattacharya, D., 2021. The empirical content of binary choice models.
Econometrica, 89(1), pp.457-474.

\item Bhattacharya, D. and P. Dupas, 2012. Inferring welfare maximizing
treatment assignment under budget constraints,\ Journal of Econometrics,
March 2012,167(1), 168--196.

\item Blackorby, C. and Donaldson, D., 1988. Money metric utility: A
harmless normalization?. Journal of Economic Theory, 46(1), 120-129.

\item Blundell, R. and Powell, J.L., 2009. Endogeneity in nonparametric and
semiparametric regression models, Chap 8 in Advances in Economics and
Econometrics: Theory and Applications, Eighth World Congress, Cambridge
University Press.

\item Blundell, R.W. and Powell, J.L., 2004. Endogeneity in semiparametric
binary response models. The Review of Economic Studies, 71(3), pp.655-679.

\item Chetty, R., 2009. Sufficient statistics for welfare analysis: A bridge
between structural and reduced-form methods. Annu. Rev. Econ., 1(1),
pp.451-488.

\item Chipman, J. and Moore, J., 1990. Acceptable indicators of welfare
change, consumer's surplus analysis, and the Gorman polar form. Preferences,
Uncertainty, and Optimality: Essays in Honor of Leonid Hurwicz, Westview
Press, Boulder.

\item Cohen, J. and Dupas, P., 2010. Free distribution or cost-sharing?
Evidence from a randomized malaria prevention experiment. Quarterly journal
of Economics, 125(1).

\item De Boor, C. (1978). A practical guide to splines. New York,
Springer-Verlag.

\item Deaton, A., 1984. Econometric issues for tax design in developing
countries, reprinted in \textquotedblleft The Theory of Taxation for
Developing Countries\textquotedblright . Washington, DC: World Bank, 1987.

\item Dreze, J., 1998. Distribution matters in cost-benefit analysis:
Comment on K.A. Brekke. Journal of Public Economics, 70(3), pp.485-488.

\item Dupas, P., 2014. Short-run subsidies and long-run adoption of new
health products: Evidence from a field experiment. Econometrica, 82(1),
pp.197-228.

\item Eilers, P.H. C. and Marx, B.D. (1996). Flexible Smoothing with $B$%
-splines and Penalties, \emph{Statistical Science}, 11, pp. 89-102.

\item Einav, L., Finkelstein, A. and Cullen, M.R., 2010. Estimating welfare
in insurance markets using variation in prices. The quarterly journal of
economics, 125(3), pp.877-921.

\item Feldstein MS. 1999. Tax avoidance and the deadweight loss of the
income tax. Rev. Econ. Stat, 81:674--80

\item Finkelstein, A. and Hendren, N., 2020. Welfare Analysis Meets Causal
Inference, Journal of Economic Perspectives, vol. 34, no. 4, pp. 146-67.

\item Fleurbaey, M. and Hammond, P.J., 2004. Interpersonally comparable
utility. In Handbook of utility theory (pp. 1179-1285). Springer, Boston, MA.

\item Garc\'{\i}a, J.L. and Heckman, J.J., 2022. On criteria for evaluating
social programs (No. w30005). National Bureau of Economic Research.

\item Goldberg, P.K. and Pavcnik, N., 2007. Distributional effects of
globalization in developing countries. Journal of economic Literature,
45(1), pp.39-82.

\item Hammond, P.J., 1990. Interpersonal comparisons of utility: Why and how
they are and should be made, European University Institute.

https://cadmus.eui.eu/bitstream/handle/1814/342/1990\_EUI\%20WP\_ECO%
\_003.pdf?sequence=1

\item Hastie, T., Tibshirani, R., Friedman, J.H. and Friedman, J.H., 2009.
The elements of statistical learning: data mining, inference, and prediction
(Vol. 2, pp. 1-758). New York: springer.

\item Hausman, J. 1981. Exact Consumer's Surplus and Deadweight Loss, 
\textit{The American Economic Review}, Vol. 71, No. 4, 662-676.

\item Hausman JA. 1996. Valuation of new goods under perfect and imperfect
competition. In The Economics of New Goods, eds. T. F. Bresnahan, RJ Gordon,
chap. 5. Chicago: University of Chicago Press, 209--248.

\item Hausman, J.A. and Newey, W.K., 2016. Individual heterogeneity and
average welfare. Econometrica, 84(3), pp.1225-1248.

\item Heckman, J.J. and Vytlacil, E.J., 2007. Econometric evaluation of
social programs, part I: Causal models, structural models and econometric
policy evaluation. Handbook of econometrics Vol. 6, 4779-4874.

\item Hendren, N. and Sprung-Keyser, B., 2020. A unified welfare analysis of
government policies. The Quarterly Journal of Economics, 135(3),
pp.1209-1318.

\item Herriges, J.A. and Kling, C.L., 1999. Nonlinear income effects in
random utility models. Review of Economics and Statistics, 81(1), pp.62-72.

\item Imbens, G.W. and Wooldridge, J.M., 2009. Recent developments in the
econometrics of program evaluation. Journal of economic literature, 47(1),
pp.5-86.

\item Jayachandran, S., 2014. Incentives to teach badly: After-school
tutoring in developing countries. Journal of Development Economics, 108,
pp.190-205.

\item Kitagawa, T. and A. Tetenov 2018. Who Should Be Treated? Empirical
Welfare Maximization Methods for Treatment Choice,\ Econometrica, 86(2),
591--616.

\item Kochar, A., 2005. Can targeted food programs improve nutrition? An
empirical analysis of India's public distribution system. Economic
development and cultural change, 54(1), pp.203-235.

\item Luenberger, D.G., 1969. Optimization by vector space methods. John
Wiley \& Sons.

\item Mas-Colell, A., Whinston, M.D. and Green, J.R., 1995. Microeconomic
theory (Vol. 1). New York: Oxford university press.

\item McFadden, D., 1981. Econometric models of probabilistic choice.
Structural analysis of discrete data with econometric applications, 198272.

\item Manski, Charles F. 2004. Statistical Treatment Rules for Heterogeneous
Populations,\ Econometrica, 72(4), 1221--1246.

\item Manski, C. F. 2014. Choosing size of government under ambiguity:
Infrastructure spending and income taxation. The Economic Journal, 124(576),
pp.359-376.

\item Mayshar, J., 1990. On measures of excess burden and their application.
Journal of Public Economics, 43(3), pp.263-289.

\item McFadden, D. 1981, Econometric Models of Probabilistic Choice, in
Manski and McFadden (eds.), Structural Analysis of Discrete Data with
Econometric Applications, 198-272, MIT Press, Cambridge, Mass.

\item Mirrlees, J.A., 1971. An exploration in the theory of optimum income
taxation. The review of economic studies, 38(2), pp.175-208.

\item Newey, W. K. 1987. Efficient estimation of limited dependent variable
models with endogenous explanatory variables. Journal of Econometrics, 36,
pp. 231-250.

\item Pollak, Robert A. "Welfare comparisons and situation comparisons."
journal of Econometrics 50, no. 1-2 (1991): 31-48.

\item Poterba, J.M., 1991. Is the gasoline tax regressive?. Tax policy and
the economy, 5, pp.145-164.

\item Ramsey, F.P., 1927. A Contribution to the Theory of Taxation. The
Economic Journal, 37(145), 47-61.

\item Rivers, D., and Q. H. Vuong. 1988. Limited information estimators and
exogeneity tests for simultaneous probit models. Journal of Econometrics 39,
pp. 347-366.

\item Saez E. 2001. Using elasticities to derive optimal income tax rates.
Rev. Econ. Stud. 68:205--29

\item Samuelson, P.A. 1947. Foundations of Economic Analysis, Ch. VIII,
"Welfare Economics".

\item Scitovszky, T., 1941. A note on welfare propositions in economics. The
Review of Economic Studies, 9(1), pp.77-88.

\item Sen, A.K. 1970. Collective Choice and Social Welfare. San Francisco:
Holden-Day.

\item Slesnick, D.T., 1998. Empirical approaches to the measurement of
welfare. Journal of Economic Literature, 36(4), pp.2108-2165.

\item Small, K.A. and Rosen, H.S., 1981. Applied welfare economics with
discrete choice models. Econometrica: Journal of the Econometric Society,
pp.105-130.

\item Smith, Richard J., and Richard W. Blundell. "An exogeneity test for a
simultaneous equation Tobit model with an application to labor supply."
Econometrica: journal of the Econometric Society (1986): 679-685.

\item Stern, N., 1987. The theory of optimal commodity and income taxation
in The theory of taxation for developing countries, The World Bank.

\item Stifel, D. and Alderman, H., 2006. The \textquotedblleft Glass of
Milk\textquotedblright\ subsidy program and malnutrition in Peru. The World
Bank Economic Review, 20(3), pp.421-448.

\item Varma, P.K., 2007. The great Indian middle class. Penguin Books India.

\item Tetenov, A., 2012. Statistical treatment choice based on asymmetric
minimax regret criteria. Journal of Econometrics, 166(1), pp.157-165.

\item Train, K.E., 2009. Discrete choice methods with simulation. Cambridge
university press.

\item Williams, H.C.W.L. 1977. On the Formulation of Travel Demand Models
and Economic Measures of User Benefit, Environment \& Planning A 9(3),
285-344.\pagebreak
\end{enumerate}

\section{Technical Appendix}

\textbf{Proof that CDF in Theorem 1 is non-decreasing:}

Note that for $c^{\prime }>c\geq y$,%
\begin{eqnarray*}
&&\Pr \left[ W\left( \mathbf{p},y,\eta \right) \leq c^{\prime }\right] \\
&=&\bar{q}_{0}\left( c^{\prime }-y+p_{1},..c^{\prime }-y+p_{J},c^{\prime
}\right) \\
&=&\Pr \left[ \max \left\{ U_{1}\left( c^{\prime }-\left( c^{\prime
}-y+p_{1}\right) ,\eta \right) ,...,U_{J}\left( c^{\prime }-\left( c^{\prime
}-y+p_{J}\right) ,\eta \right) ,\eta \right\} \leq U_{0}\left( c^{\prime
},\eta \right) \right] \\
&=&\Pr \left[ \max \left\{ U_{1}\left( y-p_{1},\eta \right) ,...,U_{J}\left(
y-p_{J},\eta \right) ,\eta \right\} \leq U_{0}\left( c^{\prime },\eta
\right) \right] \\
&\geq &\Pr \left[ \max \left\{ U_{1}\left( y-p_{1},\eta \right)
,...,U_{J}\left( y-p_{J},\eta \right) ,\eta \right\} \leq U_{0}\left( c,\eta
\right) \right] \text{, since }c^{\prime }>c\text{ and }U_{0}\left( \cdot
,\eta \right) \nearrow \\
&=&\Pr \left[ \max \left\{ U_{1}\left( c-\left( c-y+p_{1}\right) ,\eta
\right) ,...,U_{J}\left( c-\left( c-y+p_{J}\right) ,\eta \right) ,\eta
\right\} \leq U_{0}\left( c,\eta \right) \right] \\
&=&\bar{q}_{0}\left( c-y+p_{1},..c-y+p_{J},c\right) \\
&=&\Pr \left[ W\left( \mathbf{p},y,\eta \right) \leq c\right] \text{.}
\end{eqnarray*}%
For $y>c^{\prime }>c$, we have%
\begin{equation*}
\Pr \left[ W\left( \mathbf{p},y,\eta \right) \leq c\right] =0=\Pr \left[
W\left( \mathbf{p},y,\eta \right) \leq c^{\prime }\right] \text{,}
\end{equation*}%
and finally, for $c^{\prime }\geq y>c$, we have%
\begin{eqnarray*}
\Pr \left[ W\left( \mathbf{p},y,\eta \right) \leq c\right] &=&0 \\
&\leq &\bar{q}_{0}\left( c^{\prime }-y+p_{1},..c^{\prime }-y+p_{J},c^{\prime
}\right) =\Pr \left[ W\left( \mathbf{p},y,\eta \right) \leq c^{\prime }%
\right] \text{.}
\end{eqnarray*}%
\smallskip

\textbf{Proof of Assertion (\ref{11})}: Note that the move from $\mathbf{p}%
^{0}$ to $\infty _{J}$ can be decomposed into that from $\mathbf{p}^{0}$ to $%
\mathbf{p}^{1}$ and then from $\mathbf{p}^{1}$ to $\infty _{J}$, i.e.%
\begin{eqnarray*}
&&CV\left( y,\mathbf{p}^{0},\infty _{J},\eta \right) \\
&=&CV\left( y,\mathbf{p}^{0},\mathbf{p}^{1},\eta \right) +CV\left( \underset{%
\neq y}{\underbrace{y+CV\left( y,\mathbf{p}^{0},\mathbf{p}^{1},\eta \right) }%
},\mathbf{p}^{1},\infty _{J},\eta \right) \\
&\neq &CV\left( y,\mathbf{p}^{0},\mathbf{p}^{1},\eta \right) +CV\left( y,%
\mathbf{p}^{1},\infty _{J},\eta \right) \text{,}
\end{eqnarray*}%
because $y+CV\left( y,\mathbf{p}^{0},\mathbf{p}^{1},\eta \right) \neq y$,
and therefore (\ref{11}) holds.\smallskip

\textbf{Quasilinear Utilities}: If $U_{j}\left( y-p_{j},\eta \right)
=h_{j}\left( \eta \right) +y-p_{j}$, and $U_{0}\left( y,\eta \right)
=h_{0}\left( \eta \right) +y$, i.e. utility is quasilinear in the numeraire,
then our normalization becomes $\mathcal{U}_{0}\left( y,\eta \right) =y$ and 
$\mathcal{U}_{j}\left( y-p_{j},\eta \right) =y+h_{j}\left( \eta \right)
-h_{0}\left( \eta \right) -p_{j}$. Then%
\begin{eqnarray}
W\left( \infty _{J},y,\eta \right) &=&\mathcal{U}_{0}\left( y,\eta \right) =y%
\text{,} \\
W\left( \mathbf{p},y,\eta \right) &=&y+\max \left\{ 0,h_{1}\left( \eta
\right) -p_{1}-h_{0}\left( \eta \right) ,...,h_{J}\left( \eta \right)
-p_{J}-h_{0}\left( \eta \right) \right\}
\end{eqnarray}%
Therefore, 
\begin{eqnarray*}
&&y+\max \left\{ 0,h_{1}\left( \eta \right) -p_{10}-h_{0}\left( \eta \right)
,...,h_{J}\left( \eta \right) -p_{J0}-h_{0}\left( \eta \right) \right\} \\
&=&W\left( \mathbf{p}^{0},y,\eta \right) \\
&=&W\left( \mathbf{p}^{1},y+CV\left( y,\mathbf{p}^{0},\mathbf{p}^{1},\eta
\right) ,\eta \right) \\
&=&\left[ 
\begin{array}{l}
y+CV\left( y,\mathbf{p}^{0},\mathbf{p}^{1},\eta \right) \\ 
+\max \left\{ 0,h_{1}\left( \eta \right) -p_{11}-h_{0}\left( \eta \right)
,...,h_{J}\left( \eta \right) -p_{J1}-h_{0}\left( \eta \right) \right\}%
\end{array}%
\right] \text{.}
\end{eqnarray*}%
Therefore,%
\begin{eqnarray*}
&&CV\left( y,\mathbf{p}^{0},\mathbf{p}^{1},\eta \right) \\
&=&y+\max \left\{ 0,h_{1}\left( \eta \right) -p_{10}-h_{0}\left( \eta
\right) ,...,h_{J}\left( \eta \right) -p_{J0}-h_{0}\left( \eta \right)
\right\} \\
&&-\left( y+\max \left\{ 0,h_{1}\left( \eta \right) -p_{11}-h_{0}\left( \eta
\right) ,...,h_{J}\left( \eta \right) -p_{J1}-h_{0}\left( \eta \right)
\right\} \right) \\
&=&W\left( \mathbf{p}^{0},y,\eta \right) -W\left( \mathbf{p}^{1},y,\eta
\right).
\end{eqnarray*}%
Thus, we have that under quasi-linear preferences,%
\begin{equation}
W\left( \mathbf{p}^{1},y,\eta \right) -W\left( \mathbf{p}^{0},y,\eta \right)
=-CV\left( y,\mathbf{p}^{0},\mathbf{p}^{1},\eta \right) \text{.}
\end{equation}%
Indeed, if utilities are quasilinear, then it follows that the social
marginal utility of income $\frac{\partial \int W\left( \mathbf{p},y,\eta
\right) dF\left( \eta \right) }{\partial y}$ equals%
\begin{equation}
\frac{\partial }{\partial y}\int \left\{ y+\max \left\{ 0,h_{1}\left( \eta
\right) -p_{1}-h_{0}\left( \eta \right) ,...,h_{J}\left( \eta \right)
-p_{J}-h_{0}\left( \eta \right) \right\} \right\} dF\left( \eta \right) =1%
\text{,}
\end{equation}%
which does not depend on $y$

\subsection{Some details of demand estimation}

\label{sec:Appendix_details_demand}

Denote the support of $y$ as $[y_{min},y_{max}]$. Suppose we use B-splines
of degree $q$ with $M+1$ equally spaced knots on $[y_{min},y_{max}]$
(including the end points). Each of the end points in the systems of knots
enters with the multiplicity $q$. Then we have $M+q$ base B-splines which we
denote as 
\begin{equation*}
\mathcal{R}_{m,M+q}(\cdot ;q),\quad m=1,\ldots ,M+q.
\end{equation*}%
The definition of base B-splines through a recursive formula can be found,
e.g., in de Boor 1978.

To incorporate covariates $x$ in the demand estimation, we collect all the
covariates $x$ into an index $x^{\prime }\beta $ and specify the demand
function as 
\begin{equation*}
P\left( q=1\,|\,p,y,x\right) =\Phi \left( +x^{\prime }\beta \right) ,
\end{equation*}%
where $\Phi $ is the C.D.F. of the standard normal distribution. This
specification automatically imposes normalization constraints on probability
(that is, the probability varying between 0 and 1). Moreover, a monotonicity
of this function in a direction of $(p,y)$ is equivalent to the respective
monotonicity of $\beta _{p}p+\sum_{m=1}^{M}\beta _{y,m}\mathcal{R}%
_{m,M+q}(y;q)+x^{\prime }\beta $.

Using the well known formulas for the derivatives of B-splines (de Boor
1978), we obtain the following formulas for the derivative of $%
\sum_{m=1}^{M}\beta _{y,m}\mathcal{R}_{m,M+q}(y;q)$: 
\begin{align*}
\dfrac{ d \sum_{m=1}^{M}\beta _{y,m}\mathcal{R}_{m,M+q}(y;q)}{d y}
=\sum_{m=2}^M \dfrac{\beta_{y.m}-\beta_{y, m-1}}{z_{m+q}-z_{m}} \mathcal{R}%
_{m-1, M+q-1}(y; q-1),
\end{align*}
where $\mathcal{R}_{\tilde{m}, M+q-1}(\cdot; q-1)$, $\tilde{m}=1, \ldots,
M+q-1$, are base polynomials of degree $q-1$ constructed on the system of
knots $\{z_{k}\}, \; k=2,\ldots,M+2q$ (the multiplicity of each boundary
knot is now reduced by 1).

\subsection{Treatment Targeting Problem}

\label{sec:Appendix_treatment}

\subsubsection{Uniqueness and second order condition for treatment targeting
problem}

\label{sec:Appendix_uniqueness}

\textbf{Uniqueness: }In this part of the appendix, we clarify some intuition
for why the optimal subsidy problem has a unique solution and the second
order conditions for constrained maximum. It is easiest to see this in a
simple setting where $Y$ takes 2 values $y_{1}$ and $y_{2}$ w.p. $\pi _{1}$
and $\pi _{2}$ respectively. Then the optimization problem becomes%
\begin{equation*}
\max_{\sigma _{1},\sigma _{2}}\pi _{1}B\left( \sigma _{1},y_{1}\right) +\pi
_{2}B\left( \sigma _{2},y_{2}\right)
\end{equation*}%
s.t.%
\begin{equation*}
\pi _{1}C\left( \sigma _{1},y_{1}\right) +\pi _{2}C\left( \sigma
_{2},y_{2}\right) =M
\end{equation*}

To investigate uniqueness, first consider indifference maps for the benefit
and cost functions. For benefits: let $B^{\prime }\left( \cdot \right) $ and 
$B^{\prime \prime }\left( \cdot \right) $ denote 1st and 2nd price
derivative w.r.t. $\sigma $, then%
\begin{eqnarray*}
\pi _{1}B\left( \sigma _{1},y_{1}\right) +\pi _{2}B\left( \sigma
_{2},y_{2}\right) &=&const. \\
&\Rightarrow &\pi _{1}B^{\prime }\left( \sigma _{1},y_{1}\right) d\sigma
_{1}+\pi _{2}B^{\prime }\left( \sigma _{2},y_{2}\right) d\sigma _{2}=0 \\
&\Rightarrow &\left. \frac{d\sigma _{2}}{d\sigma _{1}}\right\vert _{B}=-%
\frac{\pi _{1}B^{\prime }\left( \sigma _{1},y_{1}\right) }{\pi _{2}B^{\prime
}\left( \sigma _{2},y_{2}\right) }<0 \\
&\Rightarrow &\left. \frac{d^{2}\sigma _{2}}{d\sigma _{1}^{2}}\right\vert
_{B}=-\frac{\pi _{1}B^{^{\prime \prime }}\left( \sigma _{1},y_{1}\right) }{%
\pi _{2}B^{\prime }\left( \sigma _{2},y_{2}\right) }
\end{eqnarray*}

Similarly, for costs, $C\left( \sigma _{1},y_{1}\right) =\sigma _{1}\times
q\left( \bar{p}-\sigma _{1},y_{1}\right) $, let $C^{\prime }\left( \sigma
_{1},y_{1}\right) $ and $C^{\prime \prime }\left( \sigma _{1},y_{1}\right) $
denote 1st and 2nd derivative w.r.t. $\sigma _{1}$. For indifference maps%
\begin{eqnarray*}
\pi _{1}C\left( \sigma _{1},y_{1}\right) +\pi _{2}C\left( \sigma
_{2},y_{2}\right) &=&const. \\
&\Rightarrow &\frac{d\sigma _{2}}{d\sigma _{1}}=-\frac{\pi _{1}C^{\prime
}\left( \sigma _{1},y_{1}\right) }{\pi _{2}C^{\prime }\left( \sigma
_{2},y_{2}\right) }<0 \\
&\Rightarrow &\left. \frac{d^{2}\sigma _{2}}{d\sigma _{1}^{2}}\right\vert
_{C}=-\frac{\pi _{1}C^{^{\prime \prime }}\left( \sigma _{1},y_{1}\right) }{%
\pi _{2}C^{\prime }\left( \sigma _{2},y_{2}\right) }
\end{eqnarray*}

For the application, from the net benefit curves (i.e. benefit minus cost
curves plotted against $\sigma $) we can obtain that%
\begin{equation}
B^{\prime }\left( \sigma _{2},y_{2}\right) -C^{\prime }\left( \sigma
_{2},y_{2}\right) <0\text{; \ }B^{^{\prime \prime }}\left( \sigma
_{1},y_{1}\right) -C^{^{\prime \prime }}\left( \sigma _{1},y_{1}\right) <0.
\end{equation}
From the benefit curves themselves we can get that $B^{\prime }\left( \sigma
_{2},y_{2}\right) >0$, $B^{^{\prime \prime }}\left( \sigma _{1},y_{1}\right)
>0$.

Putting all of this together, we have that%
\begin{eqnarray*}
0 &<&B^{\prime }\left( \sigma _{2},y_{2}\right) <C^{\prime }\left( \sigma
_{2},y_{2}\right) \\
0 &<&B^{^{\prime \prime }}\left( \sigma _{1},y_{1}\right) <C^{^{\prime
\prime }}\left( \sigma _{1},y_{1}\right)
\end{eqnarray*}%
Therefore, from%
\begin{equation*}
\left. \frac{d^{2}\sigma _{2}}{d\sigma _{1}^{2}}\right\vert _{B}=-\frac{\pi
_{1}B^{^{\prime \prime }}\left( \sigma _{1},y_{1}\right) }{\pi _{2}B^{\prime
}\left( \sigma _{2},y_{2}\right) }\text{ and }\left. \frac{d^{2}\sigma _{2}}{%
d\sigma _{1}^{2}}\right\vert _{C}=-\frac{\pi _{1}C^{^{\prime \prime }}\left(
\sigma _{1},y_{1}\right) }{\pi _{2}C^{\prime }\left( \sigma
_{2},y_{2}\right) }\text{,}
\end{equation*}%
we conclude that both indifference curves are decreasing and concave when
viewed from the origin, with the benefit indifference curve less concave
than the constant cost curve. Therefore there is a single internal maxima at
the point where the benefit indifference curve is tangent to the cost
indifference curve.

\textbf{Second-order conditions}: The ideas behind the second-order
conditions for this case will extend to a general case. We write the
Lagrangian as 
\begin{equation*}
\pi _{1}B\left( \sigma _{1},y_{1}\right) +\pi _{2}B\left( \sigma
_{2},y_{2}\right) +\lambda \left( M-\pi _{1}C\left( \sigma _{1},y_{1}\right)
+\pi _{2}C\left( \sigma _{2},y_{2}\right) \right) .
\end{equation*}

The FOC is then 
\begin{eqnarray*}
\pi _{1}B^{\prime }\left( \sigma _{1},y_{1}\right) &=&\lambda \pi
_{1}C^{\prime }\left( \sigma _{1},y_{1}\right) \text{, }\;\pi _{2}B^{\prime
}\left( \sigma _{2},y_{2}\right) =\lambda \pi _{2}C^{\prime }\left( \sigma
_{2},y_{2}\right) \\
&\Rightarrow &\frac{B^{\prime }\left( \sigma _{1},y_{1}\right) }{C^{\prime
}\left( \sigma _{1},y_{1}\right) }=\lambda =\frac{B^{\prime }\left( \sigma
_{2},y_{2}\right) }{C^{\prime }\left( \sigma _{2},y_{2}\right) }.
\end{eqnarray*}

The Bordered Hessian of the Lagrangian is {\small 
\begin{eqnarray*}
H &=&\left[ 
\begin{array}{ccc}
B^{\prime \prime }\left( \sigma _{1},y_{1}\right) -\lambda C^{\prime \prime
}\left( \sigma _{1},y_{1}\right) & 0 & -\pi _{1}C^{\prime }\left( \sigma
_{1},y_{1}\right) \\ 
0 & B^{\prime \prime }\left( \sigma _{2},y_{2}\right) -\lambda C^{\prime
\prime }\left( \sigma _{2},y_{2}\right) & -\pi _{2}C^{\prime }\left( \sigma
_{2},y_{2}\right) \\ 
-\pi _{1}C^{\prime }\left( \sigma _{1},y_{1}\right) & -\pi _{2}C^{\prime
}\left( \sigma _{2},y_{2}\right) & 0%
\end{array}%
\right] \\
&=&\left[ 
\begin{array}{ccc}
B^{\prime \prime }\left( \sigma _{1},y_{1}\right) -\frac{B^{\prime }\left(
\sigma _{1},y_{1}\right) C^{\prime \prime }\left( \sigma _{1},y_{1}\right) }{%
C^{\prime }\left( \sigma _{1},y_{1}\right) } & 0 & -\pi _{1}C^{\prime
}\left( \sigma _{1},y_{1}\right) \\ 
0 & B^{\prime \prime }\left( \sigma _{2},y_{2}\right) -\frac{B^{\prime
}\left( \sigma _{2},y_{2}\right) C^{\prime \prime }\left( \sigma
_{2},y_{2}\right) }{C^{\prime }\left( \sigma _{2},y_{2}\right) } & -\pi
_{2}C^{\prime }\left( \sigma _{2},y_{2}\right) \\ 
-\pi _{1}C^{\prime }\left( \sigma _{1},y_{1}\right) & -\pi _{2}C^{\prime
}\left( \sigma _{2},y_{2}\right) & 0%
\end{array}%
\right]
\end{eqnarray*}%
}

For our solution to FOC to be a local maximizer, we need $\det H\left( \cdot
\right) $ to be positive, i.e.%
\begin{multline*}
-\pi _{2}^{2}\left( C^{\prime }\left( \sigma _{2},y_{2}\right) \right)
^{2}\times \left( B^{\prime \prime }\left( \sigma _{1},y_{1}\right) -\frac{%
B^{\prime }\left( \sigma _{1},y_{1}\right) C^{\prime \prime }\left( \sigma
_{1},y_{1}\right) }{C^{\prime }\left( \sigma _{1},y_{1}\right) }\right) \\
-\pi _{1}^{2}\left( C^{\prime }\left( \sigma _{1},y_{1}\right) \right)
^{2}\times \left( B^{\prime \prime }\left( \sigma _{2},y_{2}\right) -\frac{%
B^{\prime }\left( \sigma _{2},y_{2}\right) C^{\prime \prime }\left( \sigma
_{2},y_{2}\right) }{C^{\prime }\left( \sigma _{2},y_{2}\right) }\right) >0.
\end{multline*}%
Sufficient conditions for this are 
\begin{equation*}
\frac{B^{\prime \prime }\left( \sigma _{1},y_{1}\right) }{B^{\prime }\left(
\sigma _{1},y_{1}\right) }<\frac{C^{\prime \prime }\left( \sigma
_{1},y_{1}\right) }{C^{\prime }\left( \sigma _{1},y_{1}\right) },\quad \frac{%
B^{\prime \prime }\left( \sigma _{2},y_{2}\right) }{B^{\prime }\left( \sigma
_{2},y_{2}\right) }<\frac{C^{\prime \prime }\left( \sigma _{2},y_{2}\right) 
}{C^{\prime }\left( \sigma _{2},y_{2}\right) }.
\end{equation*}

\subsubsection{Shape of Optimal Subsidy as Function of Income}

\label{sec:Appendix_treatmentshape}

In this section we derive conditions that determine shapes of optimal ATE
subsidies. In particular, we have a clear-cut condition of when these
optimal subsidies as functions of income will be decreasing.

Suppose income $y$ takes 2 values $y_{1}<y_{2}$ w.p. $\pi _{1}$ and $\pi
_{2} $ respectively. The expected ATE maximization problem becomes%
\begin{equation*}
\max_{\sigma _{1},\sigma _{2}}\pi _{1}q\left( \bar{p}-\sigma
_{1},y_{1}\right) +\pi _{2}q\left( \bar{p}-\sigma _{2},y_{2}\right)
\end{equation*}%
s.t.%
\begin{equation*}
\pi _{1}\sigma _{1}q\left( \bar{p}-\sigma _{1},y_{1}\right) +\pi _{2}\sigma
_{2}q\left( \bar{p}-\sigma _{2},y_{2}\right) =M\text{.}
\end{equation*}

The Lagrangian is given by%
\begin{equation*}
\pi _{1}q\left( \bar{p}-\sigma _{1},y_{1}\right) +\pi _{2}q\left( \bar{p}%
-\sigma _{2},y_{2}\right) +\lambda \left( M-\pi _{1}\sigma _{1}q\left( \bar{p%
}-\sigma _{1},y_{1}\right) -\pi _{2}\sigma _{2}q\left( \bar{p}-\sigma
_{2},y_{2}\right) \right)
\end{equation*}

yielding FOC%
\begin{eqnarray}
-\pi _{1}\frac{\partial }{\partial p}q\left( \bar{p}-\sigma
_{1},y_{1}\right) &=&\lambda \pi _{1}\left( q\left( \bar{p}-\sigma
_{1},y_{1}\right) -\sigma _{1}\frac{\partial }{\partial p}q\left( \bar{p}%
-\sigma _{1},y_{1}\right) \right) \\
-\pi _{2}\frac{\partial }{\partial p}q\left( \bar{p}-\sigma
_{2},y_{2}\right) &=&\lambda \pi _{2}\left( q\left( \bar{p}-\sigma
_{2},y_{2}\right) -\sigma _{2}\frac{\partial }{\partial p}q\left( \bar{p}%
-\sigma _{2},y_{2}\right) \right)
\end{eqnarray}

Taking ratios and using $\eta _{1}\left( \bar{p},\sigma _{1},y_{1}\right)
\equiv -\frac{\frac{\partial }{\partial p}q\left( \bar{p}-\sigma
_{1},y_{1}\right) }{q\left( \bar{p}-\sigma _{1},y_{1}\right) }>0$, we get%
\begin{eqnarray}
\frac{\frac{\partial }{\partial p}q\left( \bar{p}-\sigma _{1},y_{1}\right) }{%
\frac{\partial }{\partial p}q\left( \bar{p}-\sigma _{2},y_{2}\right) } &=&%
\frac{q\left( \bar{p}-\sigma _{1},y_{1}\right) -\sigma _{1}\frac{\partial }{%
\partial p}q\left( \bar{p}-\sigma _{1},y_{1}\right) }{q\left( \bar{p}-\sigma
_{2},y_{2}\right) -\sigma _{2}\frac{\partial }{\partial p}q\left( \bar{p}%
-\sigma _{2},y_{2}\right) }  \notag \\
&\Longrightarrow &\frac{\frac{\frac{\partial }{\partial p}q\left( \bar{p}%
-\sigma _{1},y_{1}\right) }{q\left( \bar{p}-\sigma _{1},y_{1}\right) }}{%
\frac{\frac{\partial }{\partial p}q\left( \bar{p}-\sigma _{2},y_{2}\right) }{%
q\left( \bar{p}-\sigma _{2},y_{2}\right) }}=\frac{1-\sigma _{1}\frac{\frac{%
\partial }{\partial p}q\left( \bar{p}-\sigma _{1},y_{1}\right) }{q\left( 
\bar{p}-\sigma _{1},y_{1}\right) }}{1-\sigma _{2}\frac{\frac{\partial }{%
\partial p}q\left( \bar{p}-\sigma _{2},y_{2}\right) }{q\left( \bar{p}-\sigma
_{2},y_{2}\right) }}  \notag \\
&\Longrightarrow &\frac{\eta _{1}}{\eta _{2}}=\frac{1+\sigma _{1}\eta _{1}}{%
1+\sigma _{2}\eta _{2}} \\
&\Longrightarrow &\eta _{1}\left( 1+\sigma _{2}\eta _{2}\right) =\left(
1+\sigma _{1}\eta _{1}\right) \eta _{2}  \notag \\
&\Longrightarrow &\eta _{1}+\sigma _{2}\eta _{2}\eta _{1}=\eta _{2}+\sigma
_{1}\eta _{1}\eta _{2}  \notag \\
&\Longrightarrow &\left( \sigma _{2}-\sigma _{1}\right) \eta _{2}\eta
_{1}=\eta _{2}-\eta _{1}  \notag \\
&\Longrightarrow &\sigma _{2}-\sigma _{1}=\frac{\eta _{2}-\eta _{1}}{\eta
_{2}\eta _{1}}=\frac{1}{\eta _{1}}-\frac{1}{\eta _{2}}
\end{eqnarray}

The above inequality is consistent with $\sigma _{2}<\sigma _{1}$ (i.e.
subsidy is progressive) if and only if $\eta _{2}<\eta _{1}$, i.e.%
\begin{equation}
\left\vert \frac{\frac{\partial }{\partial p}q\left( \bar{p}-\sigma
_{1},y_{1}\right) }{q\left( \bar{p}-\sigma _{1},y_{1}\right) }\right\vert
>\left\vert \frac{\frac{\partial }{\partial p}q\left( \bar{p}-\sigma
_{2},y_{2}\right) }{q\left( \bar{p}-\sigma _{2},y_{2}\right) }\right\vert
\label{9}
\end{equation}%
for $y_{1}<y_{2}$ and $\sigma _{2}<\sigma _{1}$, i.e. the price sensitivity
of demand measured by $\left. \left\vert \frac{\partial }{\partial p}\ln
q\left( p,y\right) \right\vert \right\vert _{p=\tilde{p},y=\tilde{y}}$ is
decreasing in $\tilde{y}$ at fixed $\tilde{p}$ and is increasing in $\tilde{p%
}$ for fixed $\tilde{y}$. The derivations here clearly extend from two to
multiple possible values of income.

For the application, Figure \ref{fig:App2shape} depicts the absolute values
of the derivatives $\left. \left\vert \frac{\partial }{\partial p}\ln
q\left( p,y\right) \right\vert \right\vert $ as functions of incomes (for
incomes from the 3rd till the 97th percentile) at different price
percentiles (even though $\bar{p}$ is the median price, since the subsidy
can be negative, we consider percentiles above the median). These absolute
values are decreasing in $y$ and, moreover, there is an ordering of the
curves with respect to price percentiles. Therefore, the optimal ATE subsidy
cannot be increasing in $y$. Indeed, suppose it increased when we moved from 
$y_{1}$ to $y_{2}$, where $y_{2}>y_{1}$. This would mean that we would move
to a lower curve and, due to the decreasing nature of each curve, the
inequality in (\ref{9}) would hold, which leads to a contradiction.
Therefore, the optimal ATE subsidy is decreasing in income, as illustrated
in Figure \ref{fig:App2allocations}. Because of the consistent shape pattern
of curves across the incomes, similar findings apply also CASW and ACV
criteria.

\begin{figure}[tbp]
\centering
\includegraphics[width=0.65%
\textwidth]{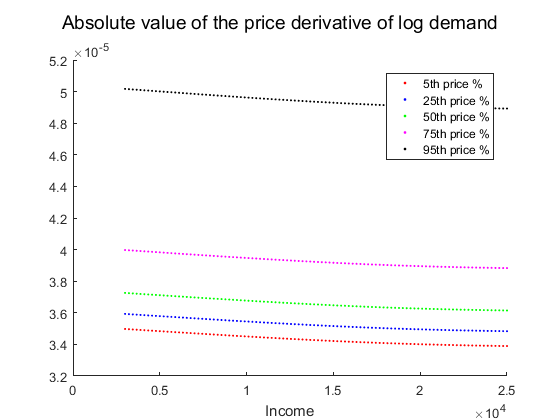}
\caption{{\protect\small Private tuition application. Absolute values of the
derivatives $\left. \left\vert \frac{\partial }{\partial p}\ln q\left(
p,y\right) \right\vert \right\vert $ as functions of incomes (for incomes
from the 3rd till the 97th percentile) at different price percentiles}.}
\label{fig:App2shape}
\end{figure}

\end{document}